\documentclass[aoas,preprint]{imsart}
\setattribute{journal}{name}{}

\newcommand{\bR}{ {\boldsymbol R} }

\newcommand{\bt}{ {\boldsymbol t} }

\newcommand{\bepsilon}{ {\boldsymbol \epsilon} }

\usepackage{amsmath}
\usepackage{pstricks,pst-grad}
\usepackage{graphicx}
\usepackage{floatrow}
\usepackage[linesnumbered,ruled,vlined]{algorithm2e}
\usepackage{subfigure}
\usepackage[utf8]{inputenc}
\usepackage{booktabs}                     
\usepackage{comment}

\usepackage{textcomp}
\usepackage{amsthm,amsmath,natbib}
\usepackage{color, colortbl}
\definecolor{Gray}{gray}{0.9}

\startlocaldefs

\endlocaldefs

\usepackage{amssymb}
\usepackage{hyperref}
\usepackage{cleveref}
\usepackage{xcolor}
\usepackage{xifthen}
\usepackage{bm}
\usepackage{bbm}

\newcommand{\R}{\mathbb{R}}
\newcommand{\x}{u}
\newcommand{\uu}{u}
\newcommand{\vv}{v}

\newcommand{\spatloc}{u}

\newcommand{\stage}{n}

\newcommand{\candidates}{\mathcal{J}}

\newcommand{\EIBV}{\operatorname{EIBV}}
\newcommand{\currentEIBV}{\operatorname{EIBV}_{[\stage]}}

\newcommand{\IBV}{\operatorname{IBV}}
\newcommand{\currentIBV}{\operatorname{IBV}_{[\stage]}}

\newcommand{\currentEEMV}{\operatorname{EEMV}_{[\stage]}}

\newcommand{\EMV}{\operatorname{EMV}}

\newtheorem*{criterion}{Criterion}

\newcommand{\T}{T}
\newcommand{\es}{\Gamma}

\newcommand{\eibv}{\mathrm{EIBV}}

\newcommand{\mes}{\nu}
\newcommand{\no}{p}
\newcommand{\domain}{\mathcal{M}}

\newcommand{\cov}{\operatorname{Cov}}

\newcommand{\gp}[1][]{
    \ifthenelse{\isempty{#1}}
    {Z}
    {Z_{#1}}
}

\newtheorem{propo}{Proposition}

\newtheorem{remark}{Remark}
\newtheorem*{remark*}{Remark}

\newcommand{\productMeasure}{
\mathrm{d}\nu^{\otimes}\left(\bm{u}\right)}

\newcommand{\covV}{C_{V}}
\newcommand{\covN}{C}

\newcommand{\jointExcuProb}{
    \mathbb{P}\left(
    \gp[\bm{u}]\in T^r \right)
}

\newcommand{\meanUU}{
    \mu(\bm{u})
}

\newcommand{\covUU}{
    K\left(\bm{u},\bm{u}\right)
}

\newcommand{\currentProba}[1]{\mathbb{P}_{[\stage]}
\left(#1\right)}
\newcommand{\futureProba}[1]{\mathbb{P}_{[\stage + 1]}
\left(#1\right)}

\newcommand{\currentExp}[1]{\mathbb{E}_{[\stage]}
\left[#1\right]}

\newcommand{\currentMean}[1]{\mu_{[\stage]}
\left(#1\right)}
\newcommand{\futureMean}[1]{\mu_{[\stage + 1]}
\left(#1\right)}

\newcommand{\currentCov}[1]{K_{[\stage]}
\left(#1\right)}
\newcommand{\futureCov}[1]{K_{[\stage + 1]}
\left(#1\right)}

\usepackage{mathtools}

\DeclarePairedDelimiterX{\expectarg}[1]{[}{]}{%
  \ifnum\currentgrouptype=16 \else\begingroup\fi
  \activatebar#1
  \ifnum\currentgrouptype=16 \else\endgroup\fi
}

\newcommand{\innermid}{\nonscript\;\delimsize\vert\nonscript\;}
\newcommand{\activatebar}{%
  \begingroup\lccode`\~=`\|
  \lowercase{\endgroup\let~}\innermid
  \mathcode`|=\string"8000
}

\begin{document}

\begin{frontmatter}

\title{
Learning Excursion Sets of Vector-valued Gaussian Random Fields for
Autonomous Ocean Sampling
}

\runtitle{Learning Excursion Sets for Autonomous Data Collection}

\begin{aug}
\author{\fnms{Trygve Olav} \snm{Fossum}\thanksref{t1,t2}, \corref{} \ead[label=e1]{trygve.o.fossum@ntnu.no}}
\author{\fnms{Cédric} \snm{Travelletti}\thanksref{t3}, \corref{} \ead[label=e2]{cedric.travelletti@stat.unibe.ch}}
\author{\fnms{Jo} \snm{Eidsvik}\thanksref{t4}, \ead[label=e3]{jo.eidsvik@ntnu.no}}
\author{\fnms{David} \snm{Ginsbourger}\thanksref{t3}, \ead[label=e4]{david.ginsbourger@stat.unibe.ch}}
\and
\author{\fnms{Kanna} \snm{Rajan}\thanksref{t5}. \ead[label=e5]{kanna.rajan@fe.up.pt}}

\affiliation[t1]{Department of Marine Technology, The Norwegian University of Science and Technology (NTNU), Trondheim, Norway.} 
\affiliation[t2]{Centre for Autonomous Marine Operations and Systems, NTNU.}
\affiliation[t3]{Institute of Mathematical Statistics and Actuarial Science, University of Bern, Switzerland.}
\affiliation[t4]{Department of Mathematical Sciences, NTNU.}
\affiliation[t5]{Underwater Systems and Technology Laboratory, Faculty of Engineering, University of Porto, Portugal.}

\address{\\Trygve Olav Fossum \\Department of Marine Technology\\ Otto Nielsens veg. 10, 7491 Trondheim\\ Norway\\
\printead{e1}}
\address{Cédric Travelletti\\ Institute of Mathematical Statistics and Actuarial Science \\ University of Bern \\
Switzerland.
\printead{e2}}
\address{Jo Eidsvik\\Department of Mathematical Sciences\\ Hogskoleringen 1, 7491 Trondheim\\ Norway\\ \printead{e3}}
\address{David Ginsbourger\\ Institute of Mathematical Statistics and Actuarial Science \\ University of Bern \\
Switzerland.
\printead{e4}}
\address{Kanna Rajan\\Underwater Systems and Technology Laboratory,
  Faculty of Engineering,\\ Rua Dr. Roberto Frias\\ University of Porto, Portugal\\
\printead{e5}}

\runauthor{TO. Fossum et al.}
\end{aug}

\begin{abstract}
  Improving and optimizing oceanographic sampling is a crucial task
  for marine science and maritime resource management. Faced with
  limited resources in understanding processes in the water-column,
  the combination of statistics and autonomous systems provide new
  opportunities for experimental design.
  In this work we develop efficient spatial sampling methods for
  characterizing regions defined by simultaneous exceedances above
  prescribed thresholds of several responses, with an application
  focus on mapping coastal ocean phenomena based on temperature and
  salinity measurements.
  Specifically, we define a design criterion based on uncertainty in
  the excursions of vector-valued Gaussian random fields, and derive
  tractable expressions for the expected integrated Bernoulli variance
  reduction in such a framework. We demonstrate how this criterion can
  be used to prioritize sampling efforts at locations that are
  ambiguous, making exploration more effective.  We use simulations to
  study and compare properties of the considered approaches, followed
  by results from field deployments with an autonomous underwater
  vehicle as part of a study mapping the boundary of a river
  plume. The results demonstrate the potential of combining
  statistical methods and robotic platforms to effectively inform and
  execute data-driven environmental sampling.
\end{abstract}


\end{frontmatter}
\section{Introduction}

Motivated by the challenges related to efficient data collection
strategies for our vast oceans, we combine spatial statistics, design
of experiments and marine robotics in this work.  The
multidisciplinary efforts enable information-driven data collection in
regions of high-interest.

\subsection{Oceanic data collection and spatial design of experiments}

Monitoring the world's oceans has gained increased importance in light
of the changing climate and increasing anthropogenic impact. Central
to understanding the changes taking place in the upper water-column is
knowledge of the bio-geophysical interaction driven by an
agglomeration of physical forcings (e.g. wind, topography, bathymetry,
tidal influences, etc.) and incipient micro-biology driven by
planktonic and coastal anthropogenic input, such as pollution and
agricultural runoff transported into the ocean by rivers.  These often
result in a range of ecosystem-related phenomena such as blooms and
plumes, with direct and indirect effects on society \citep{ryan2017}. One of the
bottlenecks in the study of such phenomena lies however in the lack of
observational data with sufficient resolution. Most of this
\emph{undersampling} can be attributed to the large spatio-temporal
variations in which ocean processes transpire, prompting the need for
effective means of data collection.  By \emph{sampling}, we refer here
primarily to the design of observational strategies in the spatial
domain with the aim to pursue measurements with high scientific
relevance.  Models and methods from spatial statistics and experimental design can clearly contribute to this sampling challenge.

Data collection at sea has typically been based on static buoys,
floats, or ship-based methods, with significant logistical limitations
that directly impact coverage and sampling resolution. Modern methods
using satellite remote-sensing provide large-scale coverage but have
limited resolution, are limited to sensing the surface, and are
impacted by cloud cover. Numerical ocean models similarly find it
challenging to provide detail at fine scale \citep{Lermusiaux:2006},
and also come with computational costs that can be limiting. The
advent of robust mobile robotic platforms \citep{Bellingham07} has
resulted in significant contributions to environmental monitoring and
sampling in the ocean (Fig.~\ref{fig:envir1}). In particular,
autonomous underwater vehicles (AUVs) have advanced the state of data
collection and consequently have made robotics an integral part of
ocean observation \citep{das11b,Das2015,fossuminformation,fossum18b}.

\begin{figure}[!h] 
  \centering 
  \subfigure[Illustration of a range of ocean sensing opportunities.]{\includegraphics[width =
    0.49\textwidth]{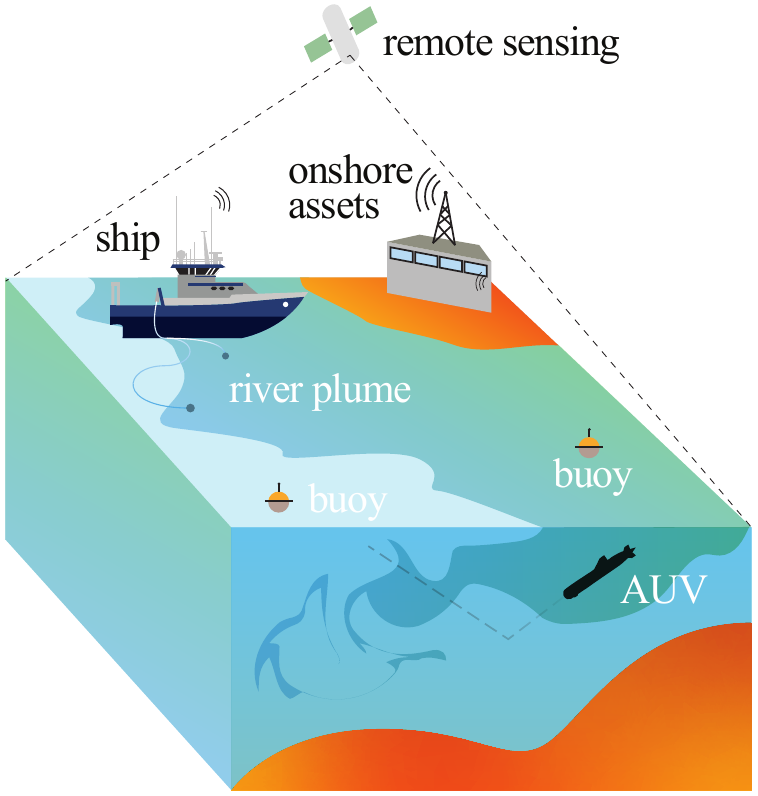}\label{fig:envir1}}
  \hfill
  \subfigure[Frontal patterns off of the Nidelva river, Trondheim, Norway.]{\includegraphics[width =
    0.49\textwidth]{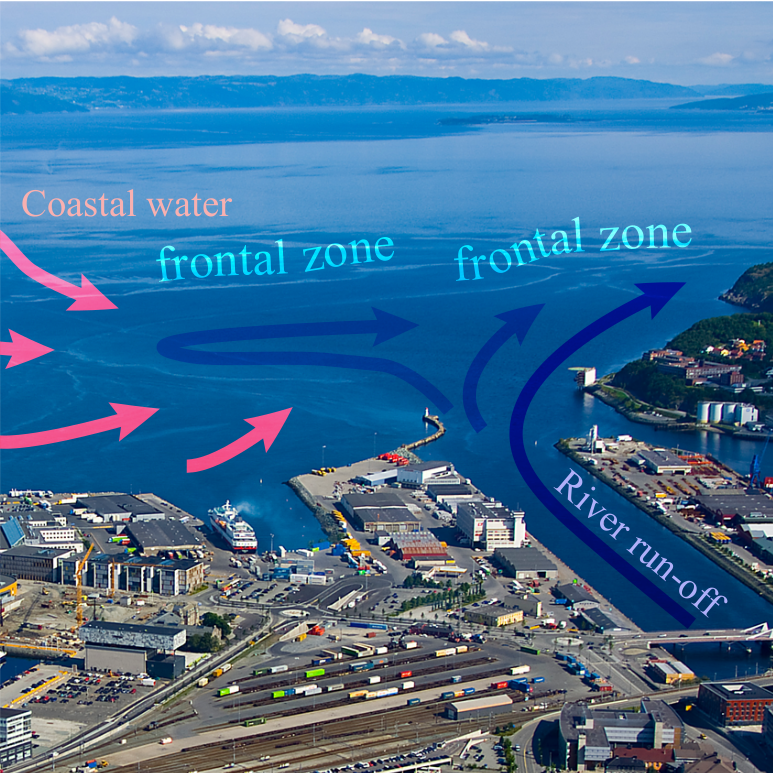}\label{fig:nidelven}}
  \caption{\ref{fig:envir1} Traditional ocean observation based on 
    ship-based sampling has been augmented by autonomous
    robotic vehicles such as AUVs. 
    \ref{fig:nidelven} The interaction of river and ocean creates
    processes that are challenging to map, where the combination of
    statistics and robotics can play a vital role in enabling more
    effective oceanographic observation.}
  \label{fig:envir} \end{figure}

Surveys with AUVs are usually limited to observations along fixed
transects that are pre-scripted in mission plans created manually by a
human operator. Missions can be specified operating on a scale of
hundreds of meters to tens of kilometers depending on the scientific
context. Faced with limited coverage capacity, a more effective
approach is to instead use onboard algorithms to continuously
evaluate, update, and refine future sampling locations, making the
data collection \emph{adaptive}.  In doing so, the space of sampling
opportunities is still limited by a waypoint graph, which forms a
discretization of the search domain where the AUV can navigate;
however the AUV can now modify its path at each waypoint based on
in-situ measurements and calculations onboard using onboard
deliberation \citep{py10,Rajan12,Rajan12b}.  Full numerical ocean
models based on complex differential equations cannot be run onboard
the AUV with limited computational capacity, and statistical models
relying on random field assumptions are relevant as a means to
effectively update the onboard model from in-situ data, and to guide
AUV data collection trajectories.

The work presented here is primarily inspired by a case study
pertaining to using an AUV for spatial characterization of a frontal
system generated by a river plume. Fig.~\ref{fig:nidelven} shows the
survey area in Trondheim, Norway, where cold freshwater enters the fjord from a river, creating a strong gradient in both temperature and
salinity. Because of the local topography and the Coriolis force the
cold fresh water tends to flow east. Depending on the variations in
river discharge, tidal effects, coastal current and wind, this
boundary often gets distorted, and knowledge about its location is
highly uncertain, making deterministic planning challenging. The goal
is therefore to use AUV measurements for improved description of the
interface between fresh and oceanic waters.  It is often not possible
to sample the biological variables of fundamental interest in such AUV
operations, but off-the-shelf instruments provide temperature and
salinity measurements which serve as proxys for the underlying
biological phenomenon. With the help of a vector-valued random field
model for temperature and salinity, one can then aim to describe the
plume.  The goal of plume characterization, in this way, relates to
that of estimating some regions of the domain, typically excursion
sets (ESs), when implicitly defined by the vector-valued random field.
In our context of environmental sampling, the joint salinity and
temperature excursions of a river plume help characterize the
underlying bio-geochemical processes
\citep{hopkins2013detection,Pinto2018}. Motivating examples for ESs of
multivariate processes are also abundant in other contexts, for
instance in medicine, where physicians do not rely solely on a single
symptom but must see several combined effects before making a
diagnosis.

The questions tackled here hence pertain to the broader area of
spatial data collection for vector-valued random fields.
Given the operational constraints on AUV movements and the fact that
surveys rely on successive measurements along a trajectory, addressing
corresponding design problems calls for sequential strategies.  Our
main research angle in the present work is to extend sequential design
strategies from the world of spatial statistics and computer
experiments to the setting of both vector-valued observational data
and experimental designs for feasible robotic trajectories. We
leverage and extend recent progress in expected uncertainty reduction
for ESs of Gaussian random fields (GRFs) in order to address this
research problem. 
We briefly review recent advances in targeted sequential design of
experiments based on GRFs before detailing other literature related to
AUV sampling and our contributions prior to outlining the rest of the
paper.

\subsection{Random field modeling and targeted sequential design of experiments}
  
While random field modeling has been one of the main topics throughout the history of spatial statistics \citep{Krige1951a,Stein1999}, even
for vector-valued random field models with associated prediction
approaches such as co-Kriging \citep[See, e.g.,][]{Wackernagel2003},
there has lately been a renewed interest for random field models in
the context of static or sequential experimental design, be it in the context of spatial data collection \citep{Mueller2007} or in
simulation experiments \citep{Santner.etal2003}. As detailed in
\cite{Ginsbourger2018}, GRF models have been used in particular as a
basis to sequential design of simulations dedicated to various goals
such as global optimization and set estimation. Of particular relevance to our context,
\cite{Bect.etal2012} focuses on strategies to reduce uncertainties related to volumes of excursion exceeding a prescribed threshold, while \cite{chevalier2014fast} concentrates on making the latter strategies computationally efficient and batch-sequential. Rather than focusing on excursion volumes, approaches were investigated in
\cite{French.Sain2013,Chevalier.etal2013b,Bolin.Lindgren2015,Azzimonti.etal2016}
with ambitions of estimating sets themselves. Recently, sequential
designs of experiments for the conservative estimation of ESs based on
GRF models were presented in \citep{Azzimonti.etal}.

Surprisingly less attention has been dedicated to sequential
strategies in the case of vector-valued observations. It has been long
acknowledged that co-Kriging could be updated efficiently in the context of sequential data assimilation \citep{Vargas-Guzman1999}, but sequential
strategies for estimating features of vector-valued random fields are
still in their infancy. \cite{LeGratiet.etal2015} used co-Kriging
based sequential designs to multi-fidelity computer codes and
\cite{Poloczek2017} used related ideas for multi-information source
optimization, but not for ES's like we do here. More relevant to our
setting, the PhD thesis \citep[][p.82]{stroh} mentions general
possibilities of stepwise uncertainty reduction strategies for ES's in
the context of designing fire simulations, yet outputs are mainly
assumed independent.

\subsection{Previous work in AUV sampling}

Other statistical work in the oceanographic domain include
\cite{wikle2013modern} focusing on hierarchical statistical models,
\cite{sahu2008space} studying spatio-temporal models for sea surface
temperature and salinity data and \cite{mellucci2018oceanic} looking
at the statistical prediction of features using an underwater glider.
In this work the main focus is not on statistical modeling per se, but
rather on statistical principles and computation underlying efficient
data collection. We combine novel possibilities in marine robotics
with spatial statistics and experimental design to provide useful AUV
sampling designs.

Adaptive in-situ AUV sampling of an evolving frontal feature has been
explored in \cite{fronts11,Smith2016,Pinto2018,costa19}. These
approaches typically use a reactive-adaptive scheme, whereby
exploration does not rely on a statistical model of the environment,
but rather adaptation is based on closing the sensing and actuation
loop. Myopic sampling, i.e. stage-wise selection of the path (on the
waypoint graph), has been used for surveys
\citep{singh2009efficient,Binney2013} that focus largely on reducing
predictive variance or entropy. These criteria are widely adopted in
the statistics literature on spatio-temporal design as well
\citep{bueso1998state,zidek2019monitoring}. However, response variance and entropy being depending only in GRF models on measurement locations and not on response values, criteria based on them only tend to have limited flexibility for
active adaptation of trajectories based on measurement values.  The use
of data-driven adaptive criteria was introduced to include more
targeted sampling of regions of scientific interest in \cite{Low2009}
and \cite{fossuminformation}.

The primary contributions of this work are:

\begin{itemize}
\item Extending uncertainty reduction criteria to
  vector-valued cases.  
\item Closed-form expressions for the expected integrated Bernoulli
  variance (IBV) of the excursions in GRFs. 
\item Algorithms for myopic and multiple-step ahead sequential
  strategies for optimizing AUV sampling with respect to the mentioned
  criteria. 
\item Replicable experiments on synthetic cases with accompanying
  code
\item Results from full-scale field trials running myopic strategies onboard an AUV for the characterization of a river plume. 
\end{itemize}

The remainder of this paper is organized as follows:
Section \ref{sec:ESEP} defines ESs, excursion probabilities (EPs), and
the design criteria connected to the IBV for excursions of
vector-valued GRFs. Section \ref{sec:heuristics} builds on these
assumptions when deriving the sequential design criteria for adaptive
sampling. In both sections properties of the methods are studied using
simulations. Section \ref{sec:case_study} demonstrates the methodology
used in field work characterizing a river plume. Section
\ref{sec:concl_disc} contains a summary and a discussion of future
work.

\section{Quantifying uncertainty on Excursion Sets implicitly defined by GRFs}
\label{sec:ESEP}

Section \ref{sec:bg_and_notation} introduces notation and co-Kriging
equations of multivariate GRFs.  Section \ref{sec:set_uq} presents
uncertainty quantification (UQ) techniques on ESs of GRFs, in
particular the IBV and the excursion measure variance (EMV).  Section
\ref{sec:eibv} turns to the effect of new observations on EMV and IBV,
and semi-analytical expected EMV and IBV over these observations are
derived.
Section \ref{Sec:UnivarEx} illustrates the concepts on a bivariate
example relevant for temperature and salinity in our case.

\subsection{Background, Notation and co-Kriging}
\label{sec:bg_and_notation}

We denote by $\gp$ a vector-valued random field indexed by some
arbitrary domain $\domain$, and assume values of the field at any
fixed location $\x \in \domain$, denoted $\gp[\x]$, to be a
$\no$-variate random vector ($\no\geq 2$). In the river plume
characterization case, $\domain$ is a prescribed domain in
Trondheimsfjord, Norway (for the purpose of our AUV application, a
discretization of a $2$-dimensional domain at fixed depth is
considered), and $\no=2$ with responses of temperature and salinity. A
bivariate GRF model is assumed for $\gp$. To motivate concepts,
Fig.~\ref{fig:real_temp} and \ref{fig:real_sal} shows a realization of such a
vector-valued GRF on $\domain=[0,1]^2$. Fig.~\ref{fig:jointex_roi} represents a by-product of interest derived from these
realizations, namely regions: i) in red, where both temperature and salinity are high (i.e., exceeding respective thresholds), indicative of ocean water, ii) in white, where both temperature and salinity are low, indicative of riverine water, and iii) in light-red, where one variable is above and the other below their respective thresholds, indicative of mixed waters.

\begin{figure}[!b] 
\centering 
\subfigure[Temperature.]{\includegraphics[
height=0.27\textwidth,keepaspectratio]{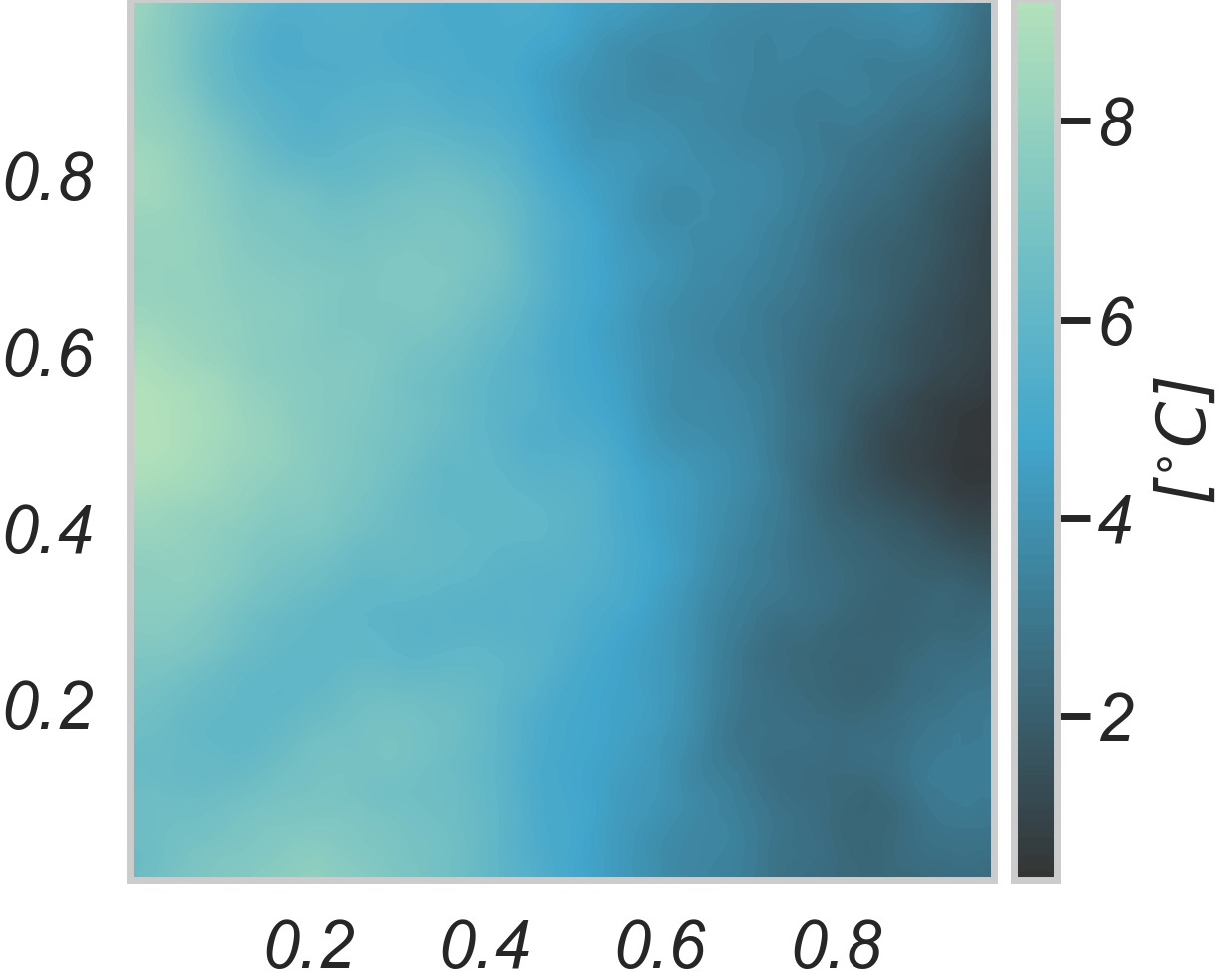}\label{fig:real_temp}}
\subfigure[Salinity.]{\includegraphics[
height=0.27\textwidth,keepaspectratio]{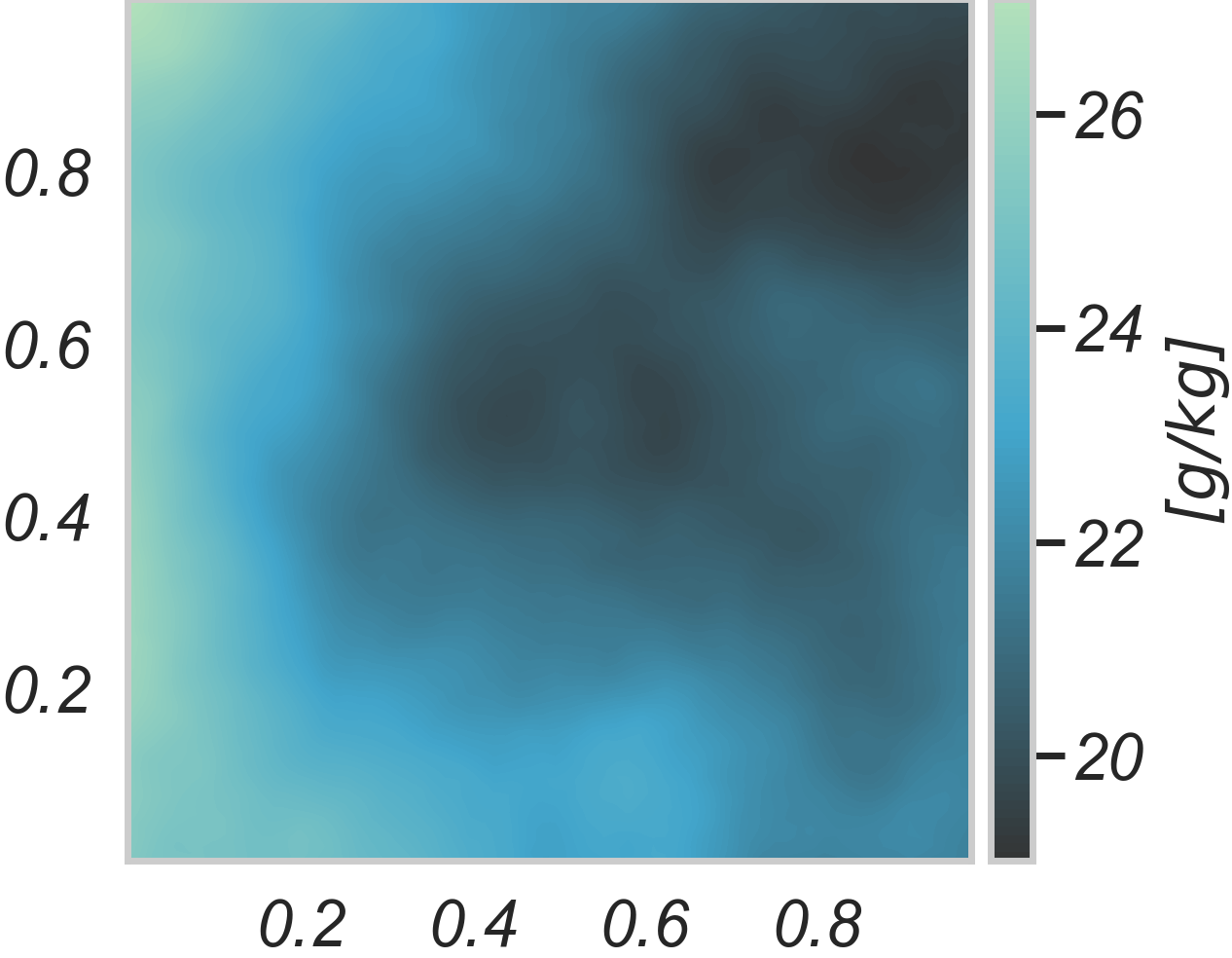}\label{fig:real_sal}}
\subfigure[Regions of interest.]{\includegraphics[
height=0.27\textwidth,keepaspectratio]{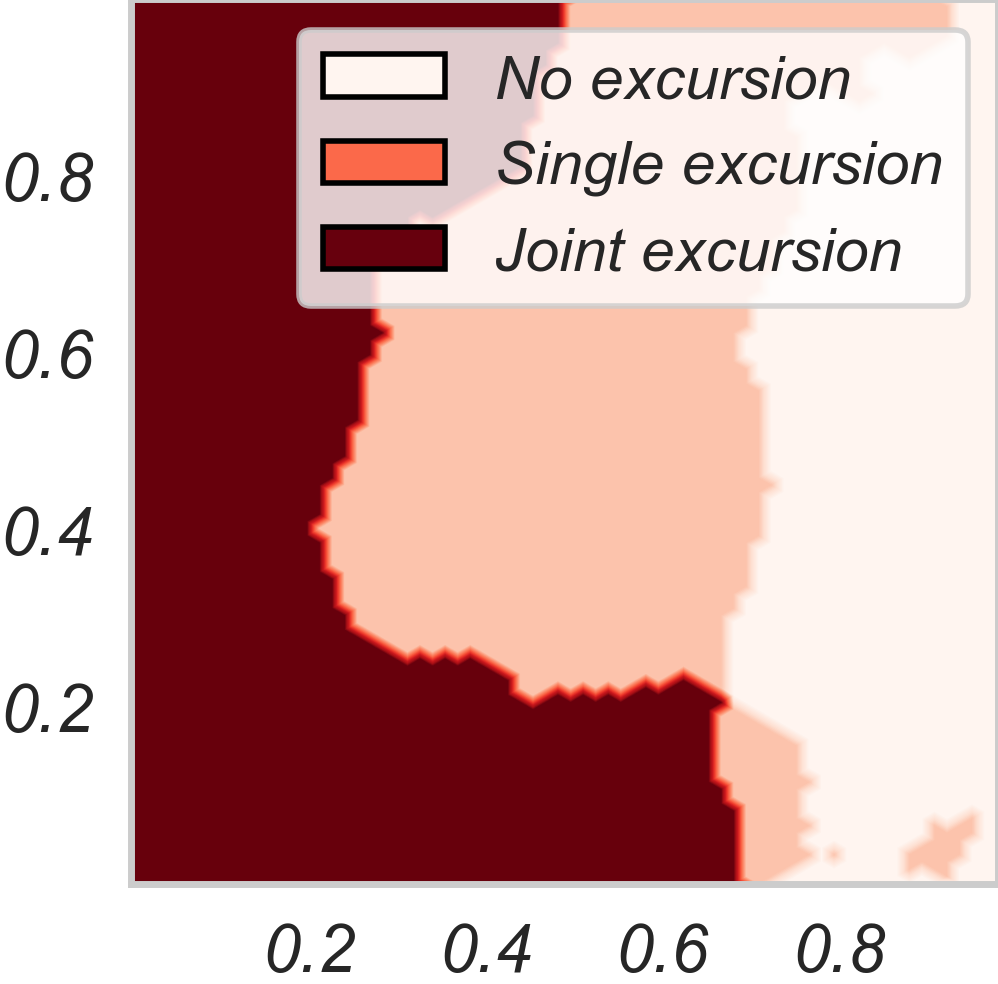}\label{fig:jointex_roi}}
\caption{Realization of a bivariate GRF (display (a) and (b)) and excursion set above some threshold (c). Joint excursion in red and excursion of a single variable in light-red.}
\label{example_excu}
\end{figure}

For the general setting of a $\no$-variate random field, we are
interested in recovering the set of locations $\es$ in the domain for
which the components of $\gp$ lie in some set of specified values
$\T\subset \mathbb{R}^{\no}$; in other words \textit{the pre-image of
  $T$ by $\gp$}:
$$
\es:=\gp^{-1}(\T)=\{\x \in \mathcal{M}: \gp[\x] \in \T\}.
$$
If we assume that $\gp$ has continuous trajectories 
and $T$ is closed, then $\es$ becomes a Random Closed Set
\citep{Molchanov2005} and concepts from the theory of random sets will
prove useful to study $\es$.
Note that while some aspects of the developed approaches do not call
for a specific form of $\T$, we will often, for purposes of
simplicity, stay with the case of orthants
($\T=(-\infty, t_1] \times \dots \times (-\infty, t_{\no}]$ where
$t_1,\dots, t_{\no} \in \R$) as this will allow efficient calculation
of several key quantities. Note that changing some $\leq$ inequalities
to $\geq$ ones would lead to immediate adaptations.

Letting $\gp[\spatloc,\ell]$ denote the $\ell\text{-th}$ component of
$\gp[\spatloc]$ ($1\leq \ell\leq \no$), we use the term
\textit{generalized location} for the couple $x=(\spatloc,\ell)$.
The notation $\gp[x]$ will be used to denote $\gp[\spatloc,\ell]$ 
and
will allow us to think of $\gp$ as a scalar-valued random field indexed by $\domain \times \{1\dots,p\}$, which will give the co-Kriging equations a particularly simple form that parallels the one of univariate Kriging. 
The letters $\spatloc$ and $\ell$ will be used for spatial locations and response indices respectively.
Furthemore, boldface letters will be used to denote concatenated
quantities corresponding to batches of observations.  Given a dataset
consisting of
$q$ observations at spatial locations
$\bm{\spatloc}=(\spatloc_1,\dots,\spatloc_q) \in
\domain^q$ and response indices $\bm{\ell}=(\ell_1,\dots, \ell_q)\in
\lbrace 1, ..., \no\rbrace^q$, 
we use the concatenated notation
\begin{align*}
\bm{x}:=
(x_1,\dots, x_q),~\text{with }x_i=(\spatloc_i,\ell_i).
\end{align*}
We also compactly denote the field values at those different locations by
\begin{align*}
\gp[\bm{x}]:=
\left(\gp[\spatloc_1,\ell_1], ...,
\gp[\spatloc_q,\ell_{q}]\right) \in \mathbb{R}^{q}.
\end{align*}

For a second order random field $(Z_{\spatloc})_{\spatloc \in
  \domain}$ with mean $\mu$ and matrix covariance function $K$,
$\mu$ is naturally extended to $\domain \in \lbrace 1, ...,
\no\rbrace$ into a function of $x=(\spatloc,
\ell)$ and is further straightforwardly vectorized into a function of
$\bm{x}$. As for $K$, it induces a covariance kernel
$k$ on the set of extended locations via $k((\spatloc,
\ell),(\spatloc', \ell'))=K(\spatloc, \spatloc')_{\ell,
  \ell'}$. In vectorized/batch form, $k(\bm{x},
\bm{x}')$ then amounts to a matrix with numbers of lines and columns equal to
the numbers of generalized locations in
$\bm{x}$ and
$\bm{x}'$, respectively. Such
vectorized quantities turn out to be useful in order to arrive at
simple expressions for the co-Kriging equations below.

Given a GRF $\gp$ and observations of some of its components at
locations in the domain, one can predict the value of the field at
some unobserved location $\spatloc\in \domain$ by using the
conditional mean of $\gp[\spatloc]$, conditional on the data. This
coincides with co-Kriging equations, which tell us precisely how to
compute conditional means and covariances.  We will present a general
form of co-Kriging, in the sense that it allows inclusion of several
(batch) observations at a time; observations at a given location
$u \in \domain$ may only include a subset of the components of
$\gp[\spatloc]\in\mathbb{R}^{\no}$ (heterotopic).

Assuming that $n$ batches of observations are available with sizes
$q_1,\dots, q_n$, and that one wishes to predict $\gp[\bm{x}]$ for
some batch of $q\geq 1$ generalized locations
$\bm{x}$, the simple co-Kriging mean then amounts to Kriging with respect to a
scalar-valued GRF indexed by $\domain\times \{1\dots,p\}$:
\begin{equation}\label{eq:cokrig_mean}
\mu_{[n]}(\bm{x})=\mu(\bm{x})+\lambda_{[n]}(\bm{x})^T (\mathbf{z}_{[n]}-\mu(\bm{x})).
\end{equation}
Here, $\mathbf{z}_{[n]}$ stands for the ($\sum_{i=1}^n
q_i$)-dimensional vector of observed (noisy) responses of
$Z$ at all considered generalized locations, and
$\lambda_{[n]}(\bm{x})$ is a vector of weights equal to
$$\left(k(\bm{x}_{[n]}, \bm{x}_{[n]})+\Delta_{[n]} \right)^{-1} k(\bm{x}_{[n]}, \bm{x})
$$
with $\bm{x}_{[n]}=(\bm{x}_1,\dots,
\bm{x}_n)$ and where
$\Delta_{[n]}$ is the covariance matrix of Gaussian-distributed noise
assumed to have affected measurements up to batch
$n$. For our applications with salinity and temperature observations,
this matrix is diagonal because we assume conditionally independent
sensor readings, but it might not be diagonal with other types of
combined measurements.  The matrix in parenthesis will be assumed to
be non-singular throughout the presentation. The associated
co-Kriging 
residual (cross-)covariance function can also be expressed in the same
vein via
\begin{equation}\label{eq:cokrig_cov}
k_{[n]}(\bm{x},\bm{x}')=k(\bm{x},\bm{x}')-\lambda_{[n]}(\bm{x})^T 
\left(k(\bm{x}_{[n]}, \bm{x}_{[n]})+\Delta_{[n]} \right)
\lambda_{{[n]}}(\bm{x}').
\end{equation}

Let us now consider the case where a co-Kriging prediction of $Z$ was
made with respect to $n$ batches of generalized locations, concatenated again within $\bm{x}_{[n]}=(\bm{x}_1,\dots, \bm{x}_n)$,
and one wishes to update the prediction by incorporating a new vector
of observations $\mathbf{z}_{n+1}$ measured at a batch of
$q_{n+1} \geq 1$ generalized locations $\bm{x}_{n+1}$.
Thanks to our representation of co-Kriging in terms of simple Kriging
with respect to generalized locations, a strightforward adaptation of
the batch-sequential Kriging update formulae from
\citep{Chevalier.etal2013a} suggests that
\begin{equation}\label{eq:meanCoK}
\mu_{[n+1]}(\bm{x})=\mu_{[n]}(\bm{x})+\lambda_{[n+1,n+1]}(\bm{x})^T (\mathbf{z}_{n+1}-\mu(\bm{x}_{n+1})),
\end{equation}
where $\lambda_{[n+1,n+1]}(\bm{x})$ denotes the $q_{n+1}$-dimensional
sub-vector extracted from $\lambda_{[n+1]}(\bm{x})$ that corresponds
to the Kriging weights for the last $q_{n+1}$ responses
when predicting at $\bm{x}$ relying on all measurements until batch $(n+1)$.
The associated co-Kriging residual (cross-)covariance function is
\begin{align}\label{eq:varCoK}
k_{[n+1]}(\bm{x},\bm{x}') & = k_{[n]}(\bm{x},\bm{x}')\\
 \nonumber - & \lambda_{[n+1,n+1]}(\bm{x})^T 
\left(k_{[n]}(\bm{x}_{n+1}, \bm{x}_{n+1})+\Delta_{n+1}\right)
\lambda_{{[n+1,n+1]}}(\bm{x}'),
\end{align}
As noted in \citep{Chevalier2015} in the case of scalar-valued fields,
these update formulae naturally extend to universal Kriging in
second-order settings and apply without Gaussian assumptions. We will
now see how the latter formulae are instrumental in deriving
semi-analytical formulae for step-wise uncertainty reduction criteria
for vector-valued random fields.

\subsection{Uncertainty Quantification on ESs of multivariate GRFs}
\label{sec:set_uq}

We now introduce quantities that allow UQ on the volume of the ES
$\es$. Let $\mes$ be a (locally finite, Borel) measure on
$\domain$. We want to investigate the probability distribution of
$\mes(\es)$ through its moments.  Centered moments may be computed
using Proposition~\ref{propo1} developed in the appendix.  In
particular, as an integral over EPs, the 
$\EMV = \operatorname{Var}[\mes(\es)]$ is:
\begin{equation*}
\begin{split}
\EMV
&=\int_{\domain^2} \mathbb{P}\left(
\gp[u]\in T, \gp[v]\in T \right)
d\mes^{\otimes}(u, v)\\
&-\left( \int_{\domain} \mathbb{P}\left(\gp[u]\in T\right) d\mes(u) \right)^2,
\end{split}
\end{equation*}
which in the excursion/sojourn case where $\T=(-\infty, t_1] \times
\dots \times (-\infty, t_{\no}]$ is
\begin{equation*}
\begin{split}
\EMV
&=\int_{\domain^2}
\varPhi_{2\no}
\left(
(\bt, \bt); \mu((u,v)),
K((u,v),(u,v))
\right)
\
\mathrm{d}\mes^{\otimes} 
(u,v)\\
&-\left( \int_{\domain} \varPhi_{\no}\left(\bt;\mu(u), K(u)\right) d\mes(u) \right)^2,
\end{split}
\end{equation*}
where $\varPhi_{\no}$ denotes the $\no$-variate Gaussian cumulative
distribution function (CDF) numerically \citep{genz2009computation}.

Note that this quantity requires the solution of an integral over
$\domain^2$. In contrast, the IBV of
\cite{bect2019} involves solely an integral on $\domain$ and can be
expanded as
\begin{equation*}
\begin{split}
\operatorname{IBV} 
&=\int_{\domain}
\mathbb{P}\left(\gp[\uu]\in T\right)(1-\mathbb{P}\left(\gp[\uu]\in T\right))
d\mes(u) \\
&=\int_{\domain}
\varPhi_{\no}\left(\bt;\mu(\uu), K(\uu)\right)
-\left(\varPhi_{\no}\left(\bt;\mu(\uu), K(\uu)\right) \right)^2
\mathrm{d}\mes(u).
\end{split}
\end{equation*}

\subsection{Expected IBV and EMV}
\label{sec:eibv}

We compute the expected effect of the inclusion of new observations on
the $\EMV$ and $\IBV$ of the ES $\es$. Let us consider the
same setting as in Eq. \eqref{eq:meanCoK} and \eqref{eq:varCoK}, and
let $\currentExp{.}$ and $\currentProba{.}$ denote conditional
expectation and probability conditional on the first $n$ batches of
observations, respectively. We use $\IBV_{\stage}$ to denote $\IBV$ with respect to the conditional law $\mathbb{P}_{\stage}$.

\medskip

In order to study the effect of the inclusion of a new data point, we
let $ \currentIBV(\bm{x}; \bm{y}) $ denote the expected IBV under the
current law of the field, conditioned on observing $\bm{y}$ at $\bm{x}$
(generalized, possibly batch observation). The expected effect of a
new observation on the IBV is then 
\begin{equation}\label{def:eibv}
    \currentEIBV(\bm{x}):=\currentExp{\mbox{IBV}(\bm{x}; \bm{Y})},
\end{equation}
where $\bm{Y}$ is distributed according to the current law of
$Z_{\bm{x}}$ and with independent noise having covariance matrix
$\Delta_n$.

We next present a result that allows efficient computation of $\EIBV$
as an integral of CDFs of the multivariate Gaussian distribution. This will prove
useful when designing sequential expected uncertainty reduction strategies.

\begin{propo}
\label{propo_eibv}
\begin{equation}
\begin{split}
\currentEIBV(\bm{x})
&=\int_{\domain} \varPhi_{\no}\left(\bt;~\currentMean{\uu}, \currentCov{u, u}\right) d\mes(u)\\
&-\int_{\domain} \varPhi_{2\no}
\left(
\left(
\begin{matrix}
\bt-\currentMean{u}\\
\bt-\currentMean{u}
\end{matrix}
\right);
\mathbf{\Sigma}_{[n]}(\uu)
\right)
d\mes(u),
\end{split}
\end{equation}
where the matrix $\mathbf{\Sigma}_{[n]}(\uu)$ is defined as
\begin{equation*}
\begin{split}
\mathbf{\Sigma}_{[n]}(\uu)&=
\left(
\begin{matrix}
\currentCov{u, u} & \currentCov{u, u}-\futureCov{u, u}\\
\currentCov{u, u}-\futureCov{u, u} & \currentCov{u, u}
\end{matrix}
\right).\\
\end{split}
\end{equation*}
\end{propo}

As for the expected EMV, a similar result may be derived.
\begin{propo}
\label{propo_emv}

\begin{equation*}
\begin{split}
\currentEEMV(\bm{x})
&=\int_{\domain^2} 
\varPhi_{2\no}
\left(
(\bt, \bt); ~ \mu((u,v)), 
K((u,v),(u,v))
\right) 
\
\mathrm{d}\mes^{\otimes} 
(u,v)\\
&-\int_{\domain^2} \varPhi_{2\no}
\left(
\left(
\begin{matrix}
\bt-\currentMean{\uu}\\
\bt-\currentMean{\vv}
\end{matrix}
\right);~
\mathbf{\tilde{\Sigma}}_{[\stage]}(\uu, \vv)
\right)
\mathrm{d}\mes^{\otimes} 
(u,v)
\end{split}
\end{equation*}
where the matrix $\mathbf{\tilde{\Sigma}}_{[n]}(\uu, \vv)$ is defined blockwise as
\begin{equation*}
\begin{split}
\mathbf{\tilde{\Sigma}}_{[n]}(\uu, \vv)&=
\left(
\begin{matrix}
\tilde{\Sigma}_{1,1}(\uu, \uu) & \tilde{\Sigma}_{1,2}(\uu, \vv)\\
\tilde{\Sigma}_{2,1}(\vv, \uu) & \tilde{\Sigma}_{2,2}(\vv, \vv)
\end{matrix}\right)
\end{split}
\end{equation*}
with blocks given, for $i,j\in \{1,2\}$ and $u,v\in \domain$, by
\begin{equation*}
\begin{split}
\tilde{\Sigma}_{i,j}(u, v) &= \lambda_{[n+1,n+1]}(\uu)^T k_{[n]}(\bm{x},\bm{x}) \lambda_{[n+1,n+1]}(\vv) + \delta_{i,j}\futureCov{\uu, \vv}.
\end{split}
\end{equation*}
\end{propo}

We remark that Propositions \ref{propo_eibv} and \ref{propo_emv} are
twofold generalizations of results from \cite{chevalier2014fast}: they
extend previous results to the multivariate setting and also allow for
the inclusion of batch or heterotopic observations through the concept
of generalized locations.  A key element for understanding these
propositions is that the conditional co-Kriging mean entering in the
EPs depend linearly on (batch) observations. The conditional equality
expressions thus become linear combinations of Gaussian variables
whose mean and covariance are easily calculated.  Related closed-form
solutions have been noted in similar contexts
\citep{bhattacharjya2013value,stroh}, but not generalized to our
situation with random sets for vector-valued GRFs.

\subsection{Expected Bernoulli variance for a two dimensional Example}
\label{Sec:UnivarEx}

We illustrate the expected Bernoulli variance (EBV) associated with
different designs on a bivariate example. This mimics our river plume
application and hence the first and second component of the random
field will be called \textit{temperature} and \textit{salinity}. We
begin with a \textit{pointwise} example, considering a single
bivariate Gaussian distribution (i.e. no spatial elements).

\subsubsection{A pointwise study}

Say we want to study the excursion probability of a bivariate Gaussian
above some threshold, where the thresholds are set equal to the mean;
$\mu_1=t_1=5^o C$ for temperature and $\mu_2=t_2=30$ g/kg for
salinity, and we play with the temperature and salinity correlation
and variances to study the effect on the EP and EBV.

\begin{figure}[!b] 
\centering
\includegraphics[width=0.99\textwidth]{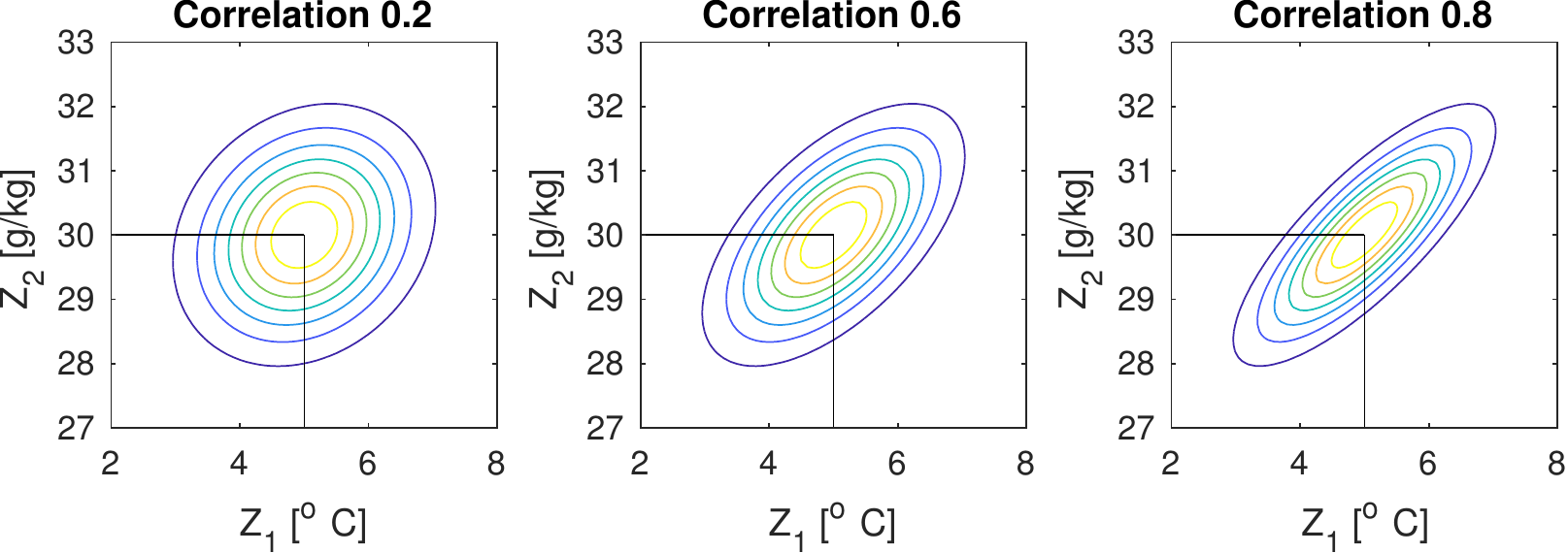}
\caption{Density contour plots with increasing correlations between
  temperature and salinity. The densities have unit variance and
  thresholds identical to the mean values $5^o C$ and
  $30$ g/kg.}
\label{illus_bivarDens}
\end{figure}

Fig.~\ref{illus_bivarDens} shows contour plots of three different
densities with increasing correlation $\gamma$ between temperature and
salinity. The displayed densities have unit standard deviations for
both temperature and salinity, but we also study the effect of
doubling the standard deviations.

Table \ref{tab:sim_rhoab} shows the initial EPs and the associated
Bernoulli variance (second row) for the examples indicated in
Fig.~\ref{illus_bivarDens}. The EPs increase with the correlation as
there is a strong tendency to have jointly low or high temperature and
salinity. The Bernoulli variance is similarly larger for high
correlations. EPs and Bernoulli variances are the same for temperature
and salinity standard deviations $\sigma_1$ and $\sigma_2$, which
implies that high variability in temperature and salinity is not
captured in the $p(1-p)$ expression.

\begin{table}[!t] \centering \caption{EP and Bernoulli variance for
    different correlations and variances (top rows), and EBVs for both
    temperature and salinity data, and only temperature data (bottom
    rows).}
  \begin{tabular}{c|ccc|ccc}
 &\multicolumn{3}{c}{$\sigma_1=\sigma_2=1$} & \multicolumn{3}{c}{$\sigma_1=\sigma_2=2$} \\
\hline
Correlation $\gamma$ & 0.2 & 0.6 & 0.8 & 0.2 & 0.6 & 0.8 \\
\hline
$p$ & 0.28 & 0.35 & 0.40 & 0.28 & 0.35 & 0.40 \\ 
$p(1-p)$ & 0.20 & 0.23 & 0.24 & 0.20 & 0.23 & 0.24 \\ 
EBV, Temperature and Salinity & 0.092 & 0.089 & 0.085 & 0.052 & 0.051 & 0.049 \\ 
EBV, Temperature only & 0.151 & 0.138 & 0.123 & 0.137 & 0.114 & 0.093 \\ 
\hline
\end{tabular}
\label{tab:sim_rhoab}
\end{table}
The bottom two rows of Table \ref{tab:sim_rhoab} show EBV
results. This is presented for a design gathering both data types, and
for a design with temperature measurements alone. When both data are
gathered, the measurement model is
$(Y_1,Y_2)^t=(Z_1,Z_2)^t+\bepsilon$, with
$\bepsilon \sim N(0,0.5^2I_2)$, while $Y_1=Z_1+\epsilon$,
$\epsilon \sim N(0,0.5^2)$ when only temperature is measured.  For
this illustration, Table \ref{tab:sim_rhoab} shows that the expected
Bernoulli variance gets smaller with larger standard deviations. The
expected reduction of Bernoulli variance is further largest for the
cases with high correlation $\gamma$. Albeit smaller, there is also
uncertainty reduction when only temperature is measured (bottom row),
especially when temperature and salinity are highly correlated. When
correlation is low ($\gamma=0.2$), there is little information about
salinity in the temperature data, and therefore less uncertainty reduction.

\subsubsection{Including Spatiality}
\label{sec:including_spatiality}

We now turn to an example involving a full-fledged GRF. The statistical model we consider has a linear trend
\begin{align*}
\mu(s)=\mathbb{E}\left[\begin{pmatrix}
Z_{\spatloc, 1}\\ Z_{\spatloc, 2}
\end{pmatrix}\right] &= \beta_0 + \beta_1 \spatloc,
\end{align*}
with $\beta_0$ a two dimensional vector and $\beta_1$ a $2\times 2$ matrix. In our examples, we only consider separable covariance models;
\begin{align*}
\textrm{Cov}\left(Z_{\spatloc, i}, Z_{v, j}\right) &= k(\spatloc, v) \gamma(i, j),~ \gamma(i, j) = \begin{cases} \sigma_i^2,~ i=j\\
   \gamma \sigma_i \sigma_j,~i\neq j,
        \end{cases}
\end{align*}
where an isotropic Mat\'{e}rn 3/2 kernel $(1+\eta h)\exp (-\eta h)$ is
used, for Euclidean distance $h$.  In the accompanying Python examples
taking place within the MESLAS
toolbox~\footnote{\url{https://github.com/CedricTravelletti/MESLAS}},
these modeling assumptions can however be relaxed to anisotropic
covariance and changing variance levels across the spatial
domain. Both extensions are relevant for the setting with river
plumes, but in practice this requires more parameters to be
specified. With extensive satellite data or prior knowledge from
high-resolution ocean models, one could also possibly fit more complex
multivariate spatial covariance functions
\citep{gneiting2010matern,genton2015cross}, but that is outside the
scope of the current work.

In the rest of this section, we consider a GRF with mean and
covariance structure as above and parameters
\begin{align*}
\beta_0 = \begin{pmatrix}
5.8\\ 24.0
\end{pmatrix}, ~ \beta_1 = \begin{pmatrix}
0.0 & -4.0\\
0.0 & -3.8
\end{pmatrix},~ \sigma_1 = 2.5,~ \sigma_2 = 2.25, ~ \gamma = 0.2,
\end{align*}
and kernel parameter $\eta=3.5$.
One realization of this GRF is shown in Fig.~\ref{example_excu}.
In the computed examples, the spatial domain $\domain$ is
discretized to a set of $N$ grid locations
$\mathcal{M}_g = \{\x_i, i=1,\ldots,N \}$, where each cell has area
$\delta$; the same grid is used for the waypoint graph for possible
design locations. The EIBV is approximated by sums over all grid
cells.

We now study how the EBV [Eq.\eqref{def:eibv}] associated with data
collection at a point changes if only one of the two components of the
field is observed. We first draw a realization of the GRF defined
above and use it as ground-truth to mimic the real data-collection
process. A first set of observations are done at the locations
depicted in grey (see Fig.~\ref{fig:ebv_comp}), and the data is used
to update the GRF model. We then consider the green triangle as a
potential next observation location and plot the EBV reduction (at
each grid node in the waypoint graph) that would result from observing
only one component of the field (temperature or salinity), or both at
that point.

\begin{figure}[!b] 
\centering 
\subfigure[Regions of interest.]{\includegraphics[
height=0.21\textwidth,keepaspectratio]{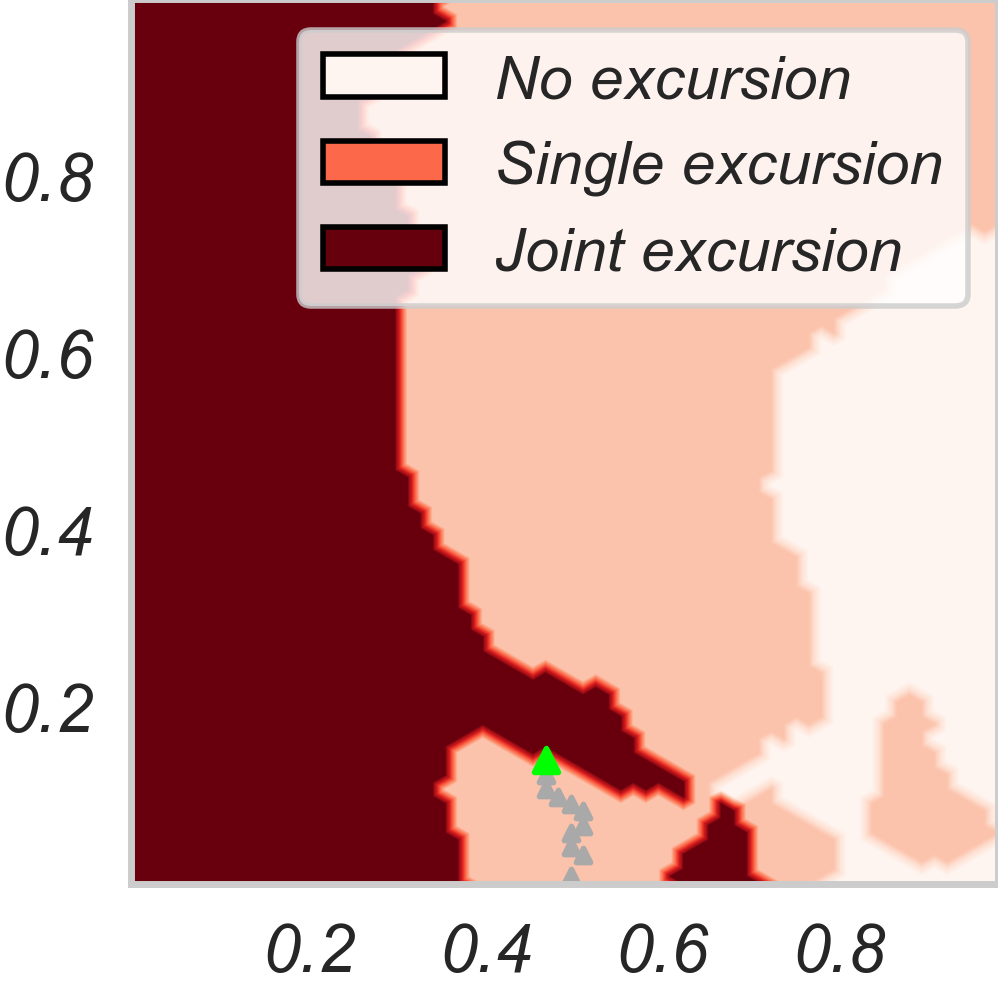}\label{fig:ebv_comp_excu}}
\subfigure[Temperature.]{\includegraphics[
height=0.22\textwidth,keepaspectratio]{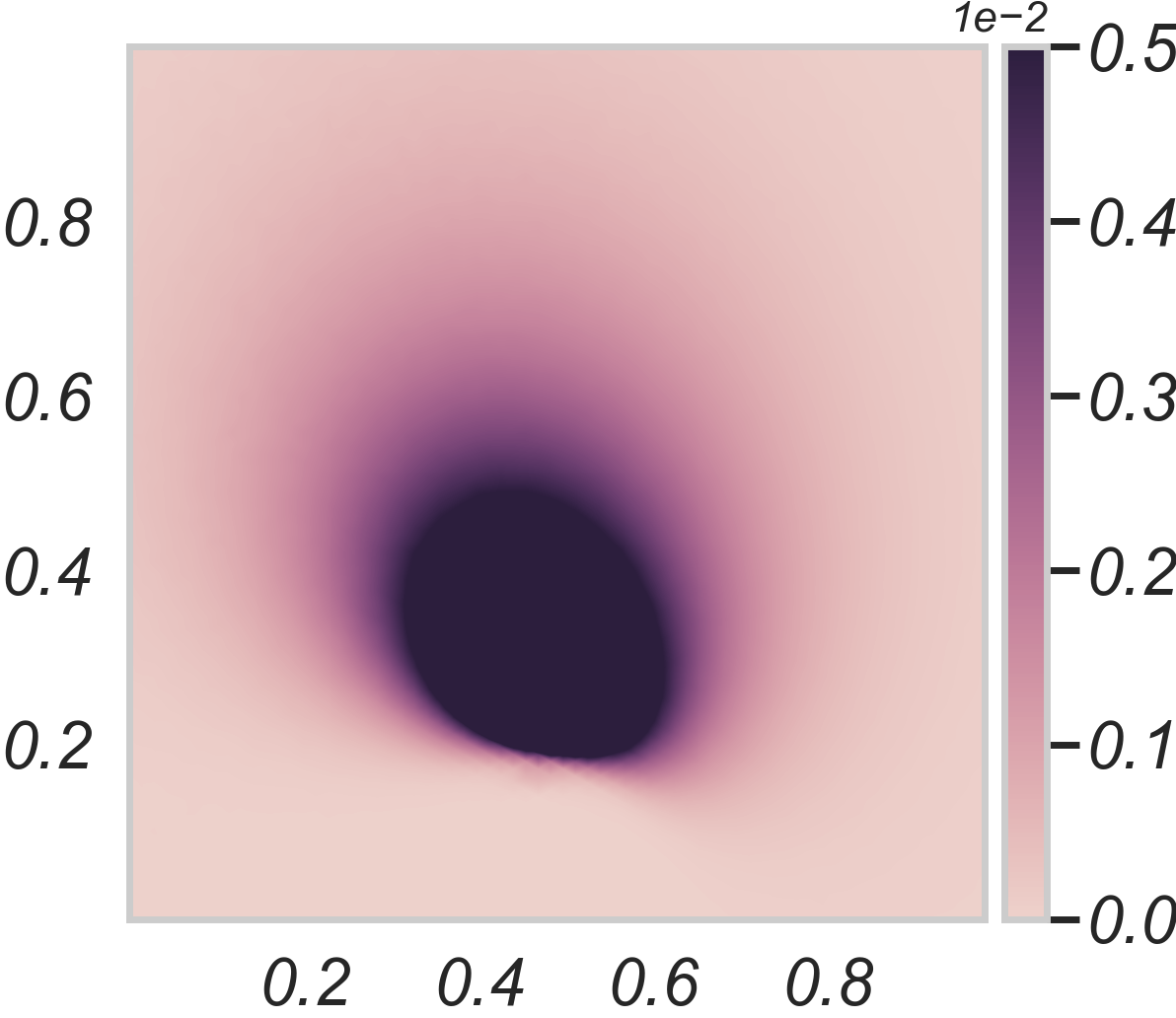}\label{fig:ebv_comp_temp}}
\subfigure[Salinity.]{\includegraphics[
height=0.22\textwidth,keepaspectratio]{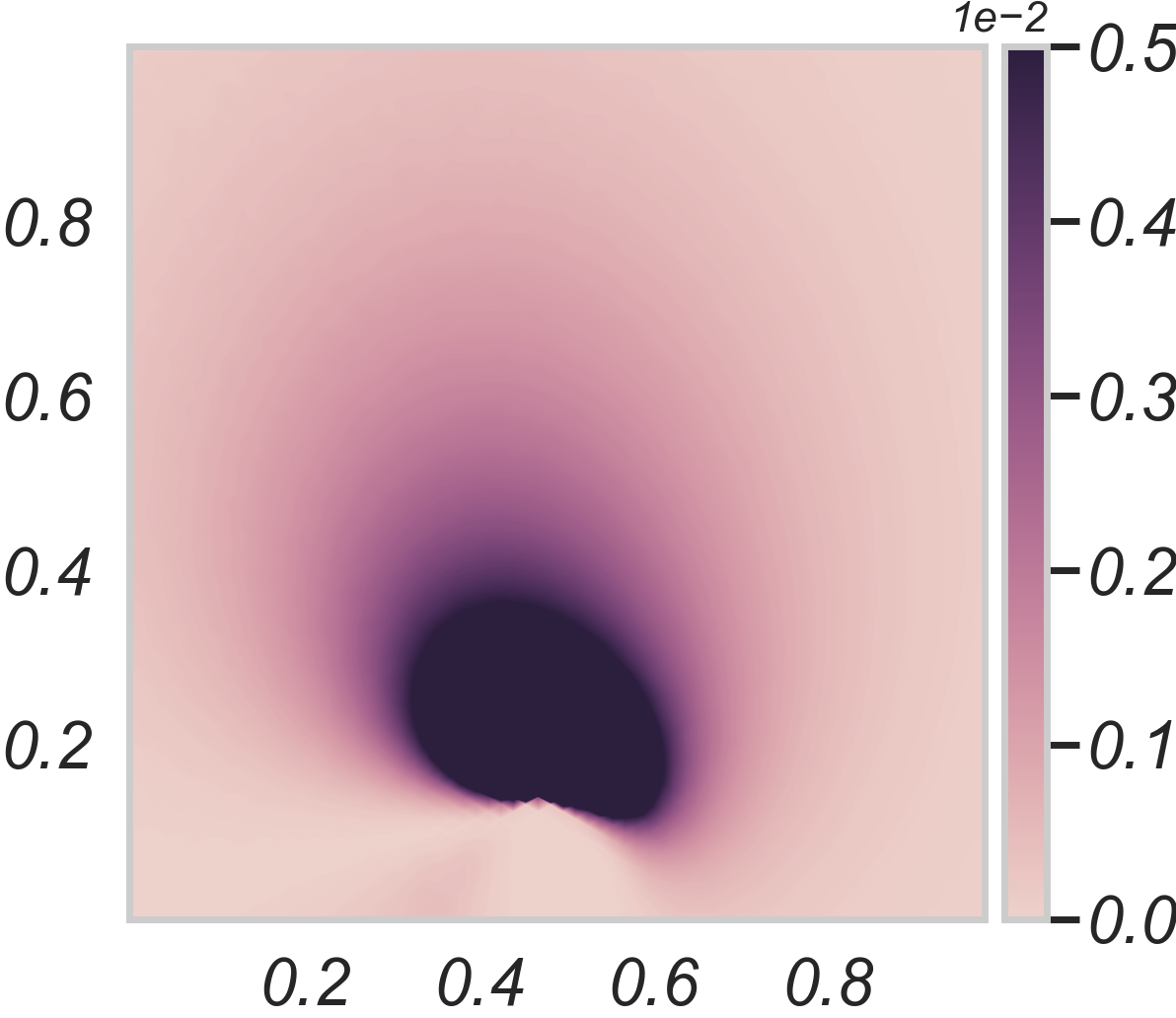}\label{fig:ebv_comp_sal}}
\subfigure[Both.]{\includegraphics[
height=0.22\textwidth,keepaspectratio]{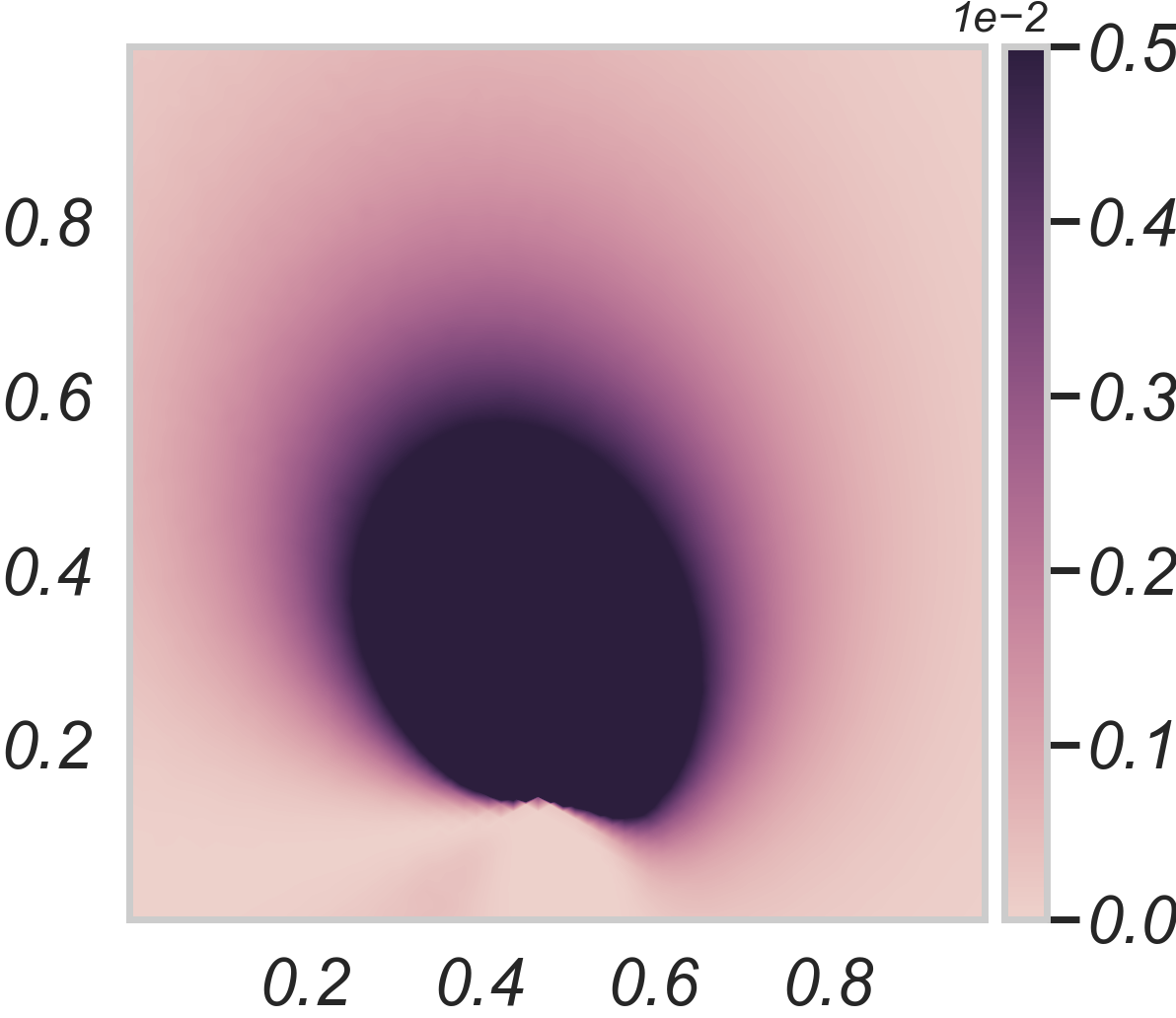}\label{fig:ebv_comp_both}}
\caption{Pointwise Bernoulli variance reduction for observation of a
  single or both components of the random field at one location. Data
  collection locations in green. True excursion set in red. Places
  where only one response is above threshold are depicted in pink. EBV
  reduction associated to observing one or both responses at the green
  location are shown in \ref{fig:ebv_comp_temp},
  \ref{fig:ebv_comp_sal} and \ref{fig:ebv_comp_both}.}
\label{fig:ebv_comp}
\end{figure}

Note that plotting the EBV reduction at each point might also be used
to compare different data collection plans. For example,
Fig. \ref{fig:ebv_north_vs_east} shows the EBV reduction associated
with a data collection plan along a vertical line (static north) and
one associated with a horizontal (static east). Both expectations are
computed according to the a-priori distribution of the GRF (i.e. no
observations have been included yet).

\begin{figure}[ht] 
\centering 
\subfigure[Excursion probability.]{\includegraphics[
height=0.25\textwidth,keepaspectratio]{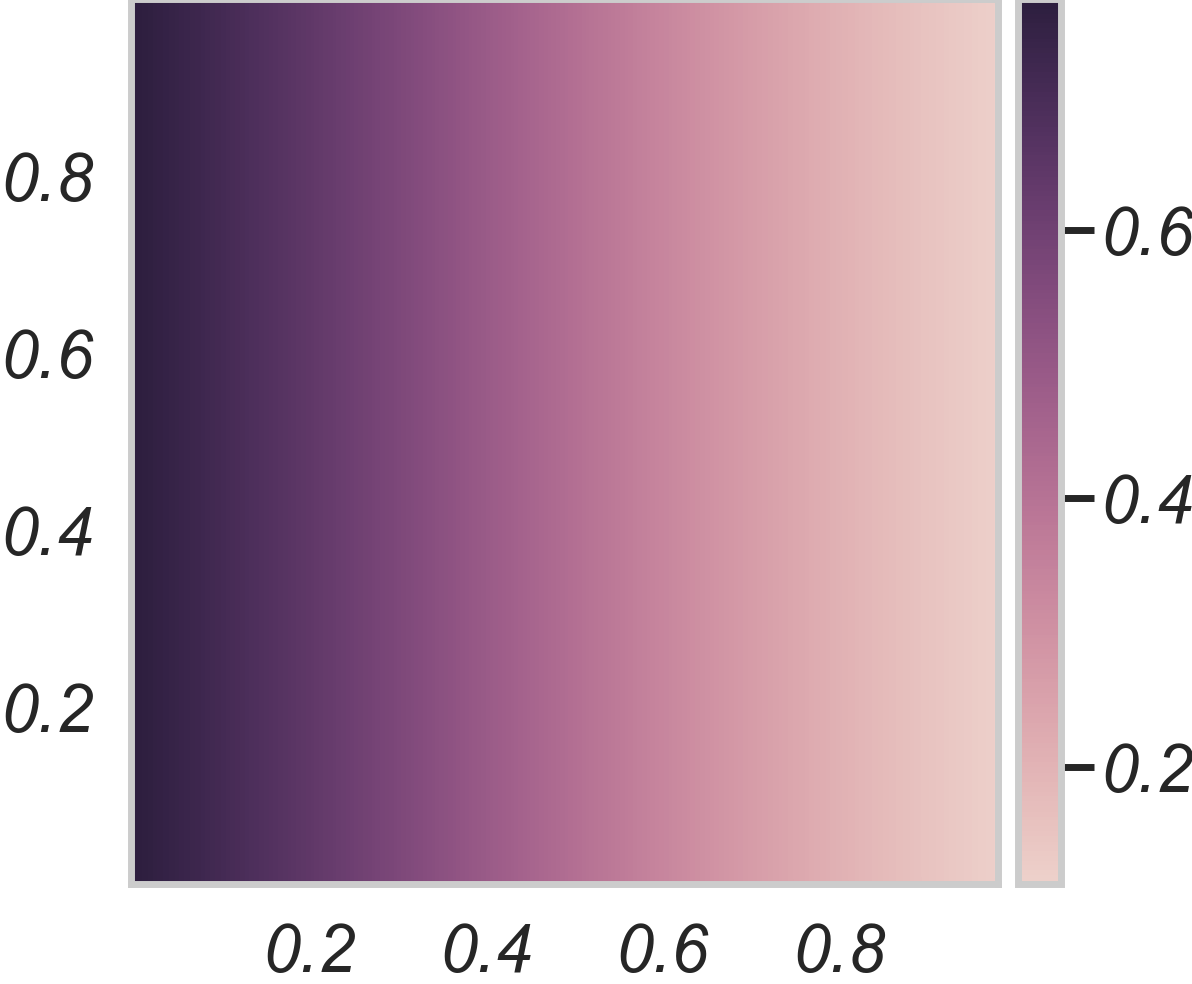}\label{fig:ebv_static_excu}}
\subfigure[Static north design.]{\includegraphics[
height=0.26\textwidth,keepaspectratio]{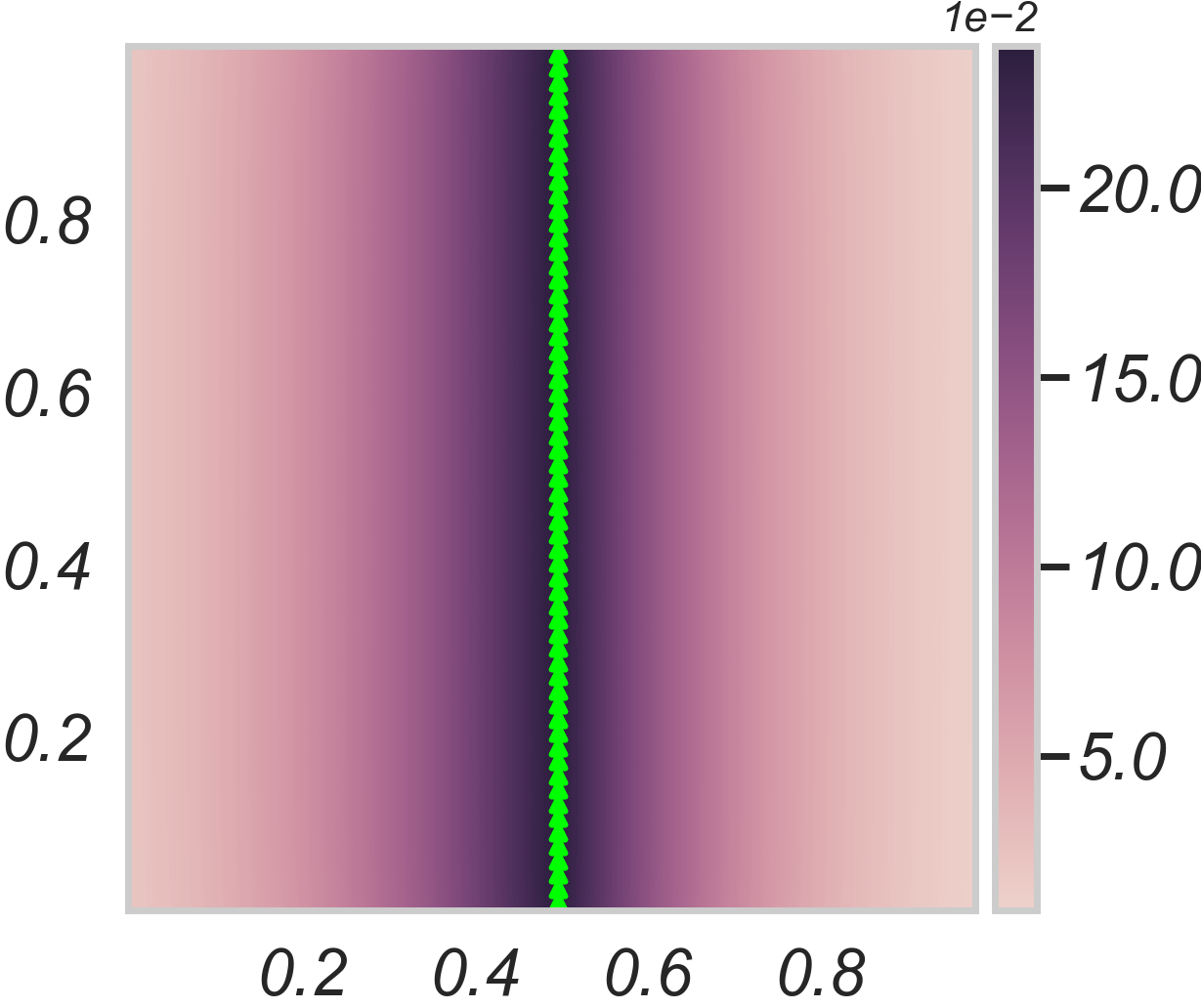}\label{fig:ebv_static_north}}
\subfigure[Static east design.]{\includegraphics[
height=0.26\textwidth,keepaspectratio]{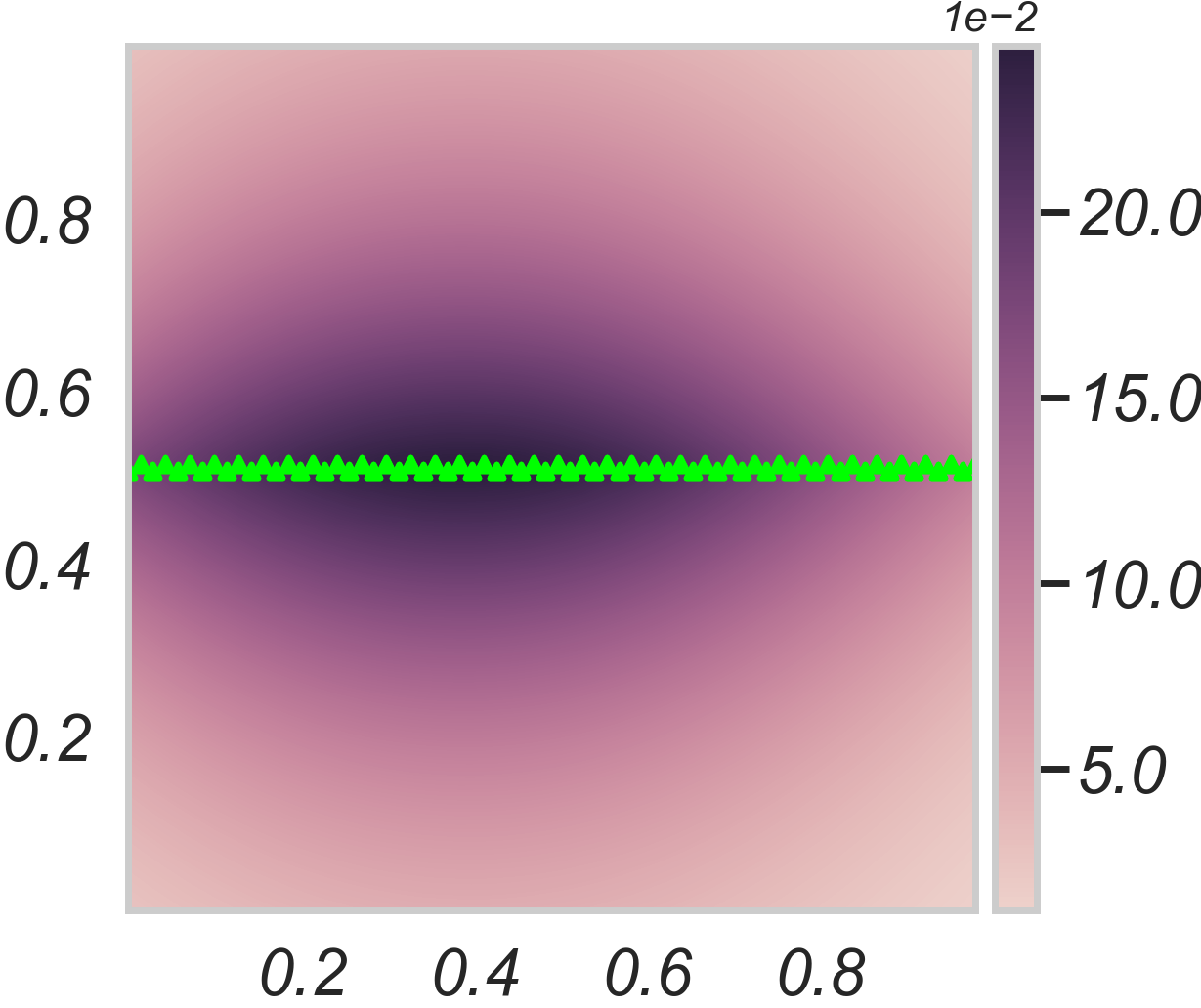}\label{fig:ebv_static_east}}
\caption{Pointwise Bernoulli variance reduction for two different static designs (later noted as \textit{static\_north} and \textit{static\_east}). The prior EP is shown in
  \ref{fig:ebv_static_excu}. EBV reduction for each design shown in
  \ref{fig:ebv_static_north} and \ref{fig:ebv_static_east}.}
\label{fig:ebv_north_vs_east}
\end{figure}

\section{Sequential designs and heuristic path planning}
\label{sec:heuristics}

We present sequential data collection
strategies that aim at reducing the expected uncertainty on the target
ES $\es$.

\subsection{Background}

From a sequential point of view, $n$ data collection steps have
already been performed and one wants to choose what data to collect
next. The design evaluations
are based on the conditional expectation $\currentExp{.}$ from the law
of the field, conditional on all data available at
stage $n$.
Once the best design at stage $n$
has been selected, the data are collected and the GRF
model is updated using co-Kriging
Eq. \ref{eq:meanCoK} and \ref{eq:varCoK}, yielding a conditional law $\mathbb{P}_{\stage + 1}$
after which the process is repeated.

Note that the type of data collected at each stage can be of various
type (all components of the field at a single location, only some
components at a subset of selected locations, etc.) because of the
concept of \textit{generalized location} in the co-Kriging expressions.
In general, a design strategy must choose the spatial location as well
as the components to observe (heterotopic), or where several
observations are allowed at each stage (batch).  For the case with an
AUV exploring the river plume, we limit our scope to choosing one of
the neighboring spatial location (waypoints) at each stage, and all
components (temperature and salinity) of the field are observed
(isotopic). The candidate points at this stage are denoted
$\candidates$ as defined from the 6 directions (apart from edges) in
the waypoint graph (see Fig.~\ref{fig:wp_graph_a}). The
set $\candidates$ depends on the current location, but for readability
we suppress this in the notation.

The mathematical expression for the optimal design in this sequential
setting involves a series of intermixed maximizations over designs and
integrals over data. In practice, the optimal solution is intractable
because of the enormous growth over stages (see
e.g. \cite{powell2016perspectives}).
Instead, we outline heuristic strategies.

\subsection{A Naive Sampling Strategy}
\label{naive}

A simple heuristic for adaptive sampling is to observe $Z$ at the
location in $\candidates$ with current EP closest to
$\frac{1}{2}$. While easy to implement, this strategy can lead to
spending many stages in boundary regions regardless of the possible
effect of sampling at the considered point for the future conditional
distribution of $Z$. The strategy does not account for the expected
reduction in uncertainty, and it does not consider having an
integrated effect over other locations.

\subsection{Myopic Path Planning}
\label{sec:myopic}

The myopic (greedy) strategy which we present here is optimal if we
imagine taking only one more stage of measurements; it does not
anticipate what the subsequent designs might offer beyond the first
stage.  Based on the currently available data the myopic strategy
selects the location that leads to the biggest reduction in EIBV:
\begin{criterion}[Myopic]
The next observation location $\spatloc_{\stage + 1}$ is chosen among
the minimizers in $\candidates$ of the criterion: 
\begin{equation}\label{critSEQ}
     C_{\text{myopic}}(u) = \EIBV_{\stage}\left(\spatloc\right)
\end{equation}
\end{criterion}

The EIBV is efficiently computed for each of the candidate points
$\candidates$ using Proposition \ref{propo2}. 
Even though this myopic strategy is non-anticipatory, it still
provides a reasonable approach for creating designs in many
applications. Moreover, it can be implemented without too much demand on
computational power, making it well-suited for embedding on an AUV.

\begin{figure}[ht] 
\centering 
\subfigure[Excursion realization.]{\includegraphics[
height=0.26\textwidth,keepaspectratio]{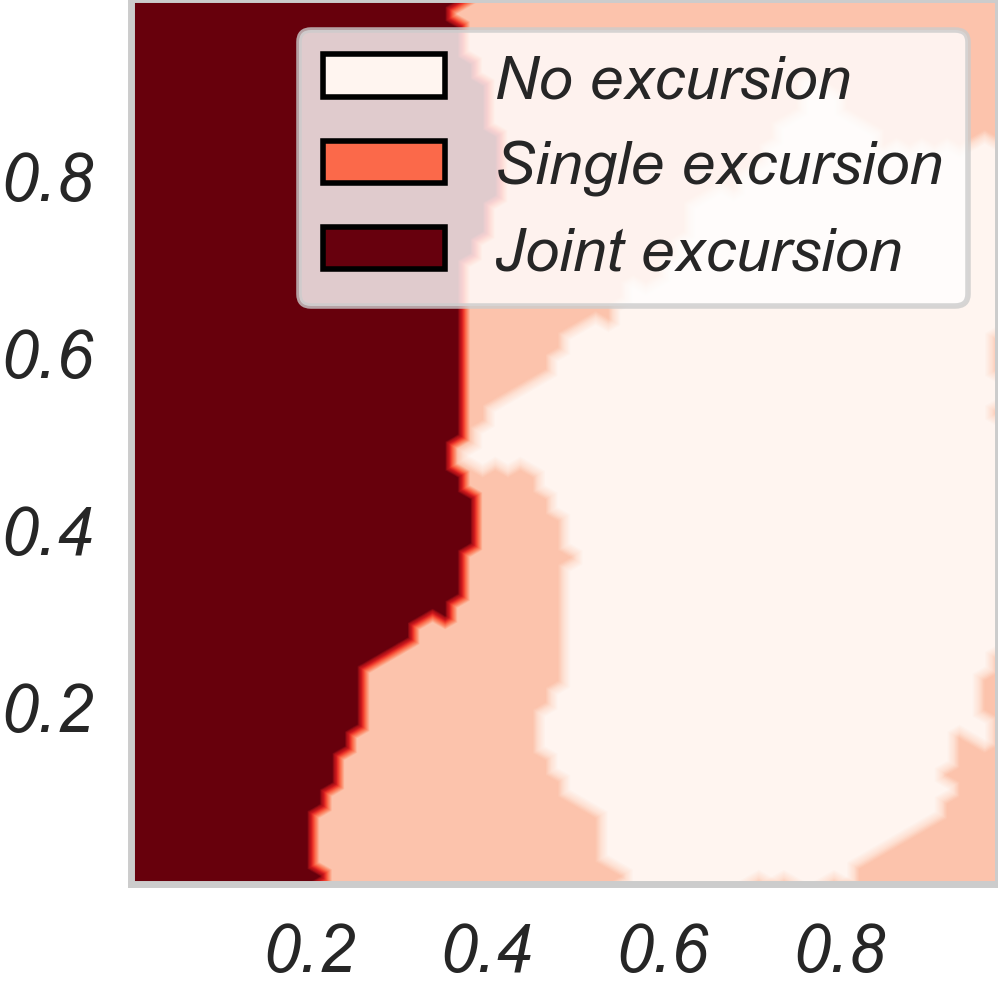}\label{fig:real_myopic}}
\hspace{2em}
\subfigure[BV reduction.]{\includegraphics[
height=0.26\textwidth,keepaspectratio]{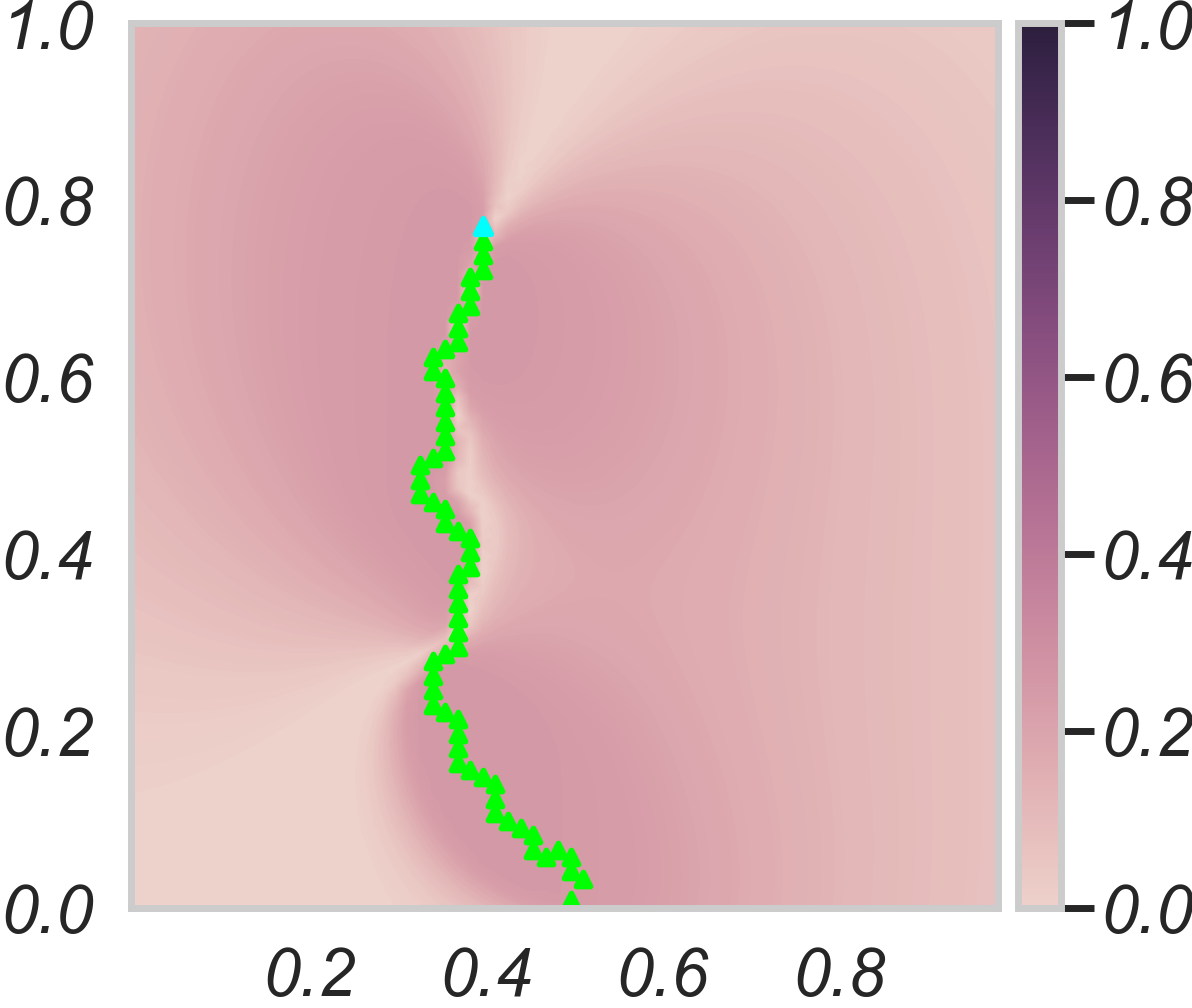}\label{fig:route_myopic}}
\hspace{2em}
\subfigure[Expected reduction of BV.]{\includegraphics[
height=0.26\textwidth,keepaspectratio]{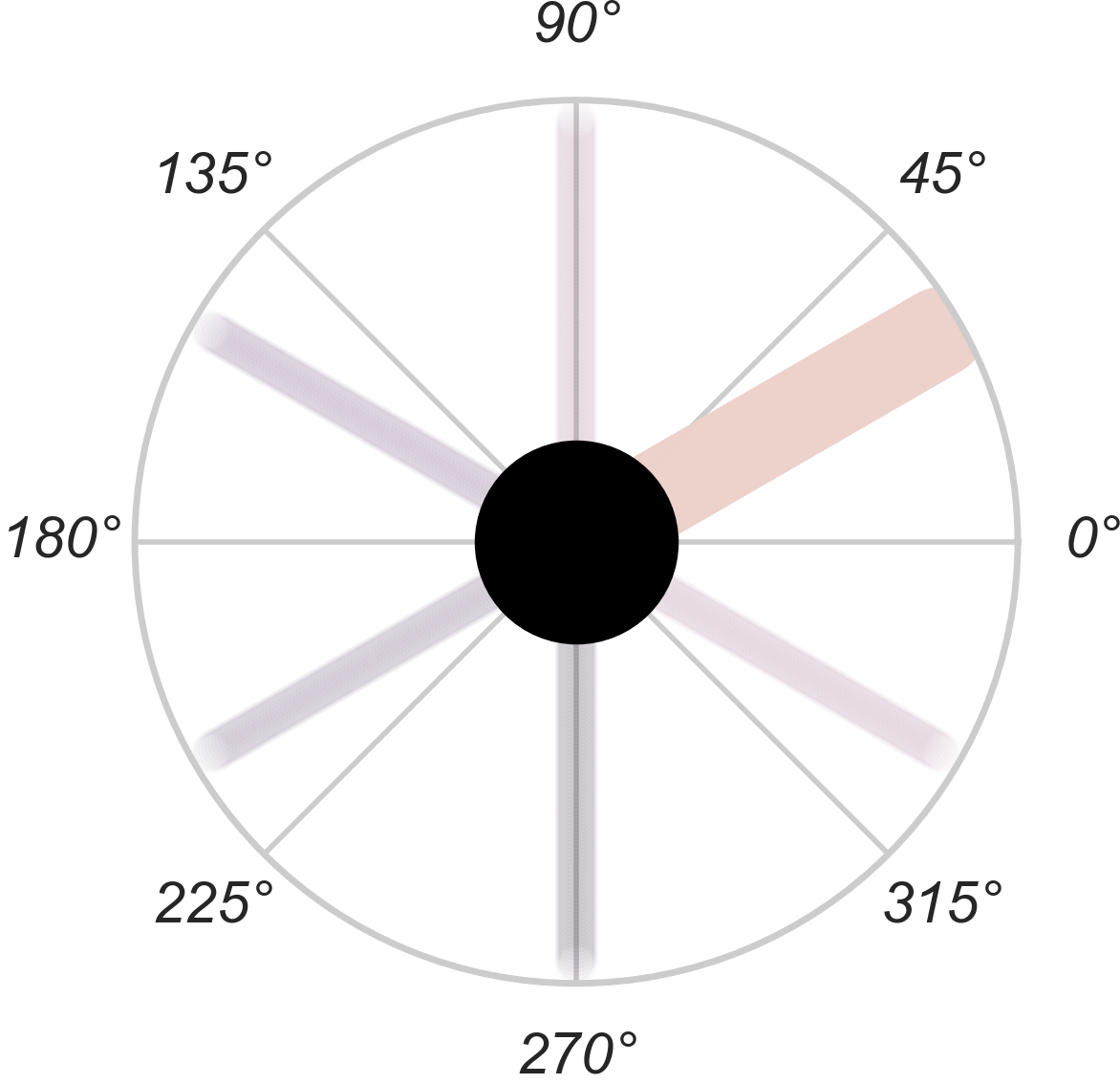}\label{fig:compass_myopic}}
\caption{Example run of the myopic strategy on a realisation of the GRF
model from \ref{sec:including_spatiality}. Reduction in Bernoulli variance
compared to the prior is shown in \ref{fig:route_myopic}, with past observation
locations in green and current AUV position in cyan. The expected IBV reduction
associated to data collection at neighbouring nodes of the current location is
shown in \ref{fig:compass_myopic}. The thick and light color indicates the node at $30 ^{\circ}$ to be the best possible choice.}
\label{fig:ebv_myopic}
\end{figure}

\subsection{Look-ahead Trajectory Planning}
\label{sec:LA}

We now extend the myopic strategy by considering two stages of
measurements, which is optimal in that it accounts consistently for
the expectations and minimizations in these two stages, but no anticipation beyond that.

The principle of
two-step look-ahead is to select the next observation location $\spatloc_{\stage + 1}$ that yields the biggest reduction in EIBV if
we were to (optimally) add one more observation after that again. In order to formalize this concept, we must
extend the notation for EIBV in the future (after observation
$\stage + 1$ has been made). We let $\currentEIBV(\cdot; u, y)$
denote the EIBV where expectations are taken conditional on the data
available at stage $n$ and on an additional observation $y$ at $u$ at stage $n+1$.

\begin{criterion}[2-step look-ahead]
      The next observation location $\spatloc_{\stage + 1}$ is chosen among the minimizers in $\candidates$ of the criterion
      \begin{align}\label{critLA}
          C_{\text{2-steps}}(u) &= \mathbb{E}_{\stage}\left[\min_{\spatloc' \in
                  \candidates(\spatloc)} \EIBV_{\stage}\left(\spatloc' ; \spatloc,
      Y\right)\right]
      \end{align}
      where $Y$ is the random data realization of $\gp_{\spatloc}$ according to its
      conditional law at step $\stage$ with the
      dependence of the set of candidates on the current location
      having been made explicit for the second stage of measurements.
    \end{criterion}
    
    In a practical setting, the first expectation can be computed by
    Monte Carlo sampling of data $Y$ from its conditional
    distribution. For each of these data samples, the second
    expectation is solved using the closed-form expressions for EIBV
    provided by Proposition \ref{propo2}, now with conditioning on the
    first stage data already going into the co-Kriging updating
    equations.

\subsection{Simulation studies}
\label{sec:simulations}

\subsubsection{Static and Sequential Sampling Designs}
\label{sec:sampling_designs}

We compare three different static designs denoted
\textit{static\_north}, \textit{static\_east}, and
\textit{static\_zigzag} (a version of \textit{static\_north} where with some east-west transitions in a zigzag pattern) with the three described sequential approaches
\textit{naive}, \textit{myopic}, and \textit{look-ahead}. The static
AUV sampling paths are pre-scripted and cannot be altered.
For a fixed survey length, a closed-form expression for the EIBV is
available as in Proposition \ref{propo_eibv}. However, for the sequential
approaches this is not the case. For comparison, the properties are
therefore evaluated using Monte Carlo integration over several replicates
of realizations from the model while conducting simulated sequential
surveys for each one. An example of such a realization with a \textit{myopic}
strategy is shown in Fig.~\ref{fig:ebv_myopic}. 

We also compare predictive
performance measured by root mean square error (RMSE) for temperature
and salinity estimates as well as the variance reduction in these
two variables. It is important to note that the objective function
used by the AUV is focused on reducing the
EIBV, but we nevertheless expect that we will achieve good predictive
performance for criteria such as RMSE as well. Another non-statistical
criterion that is important for practical purposes is the computational
time needed for the strategy.

\begin{figure}[!b] 
\centering 
\subfigure[The waypoint graph.]{\includegraphics[width =
0.49\textwidth]{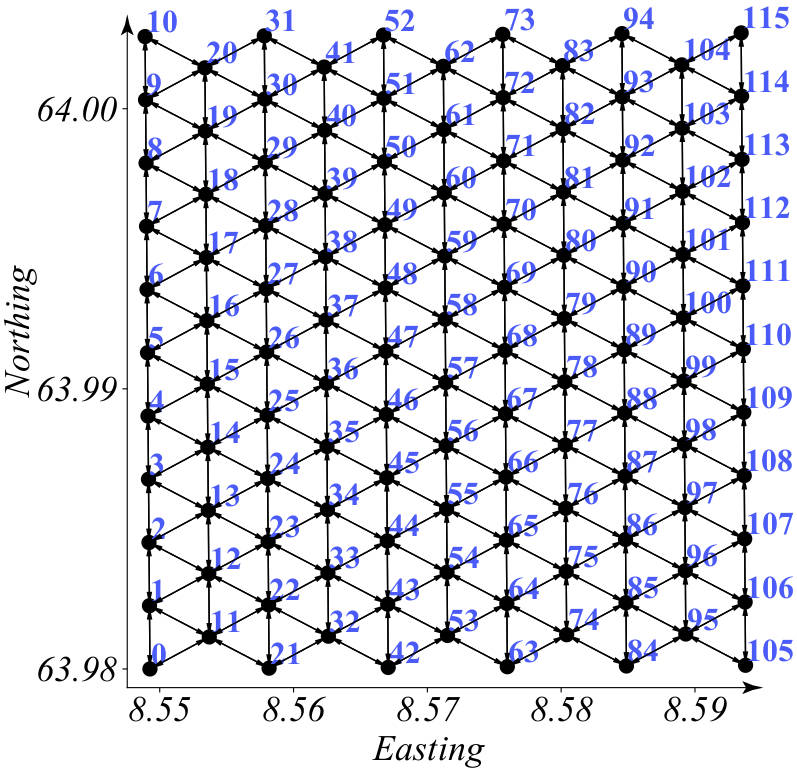}\label{fig:wp_graph_a}}
\hfill
\subfigure[The waypoint graph in 3D.]{\includegraphics[width =
0.49\textwidth]{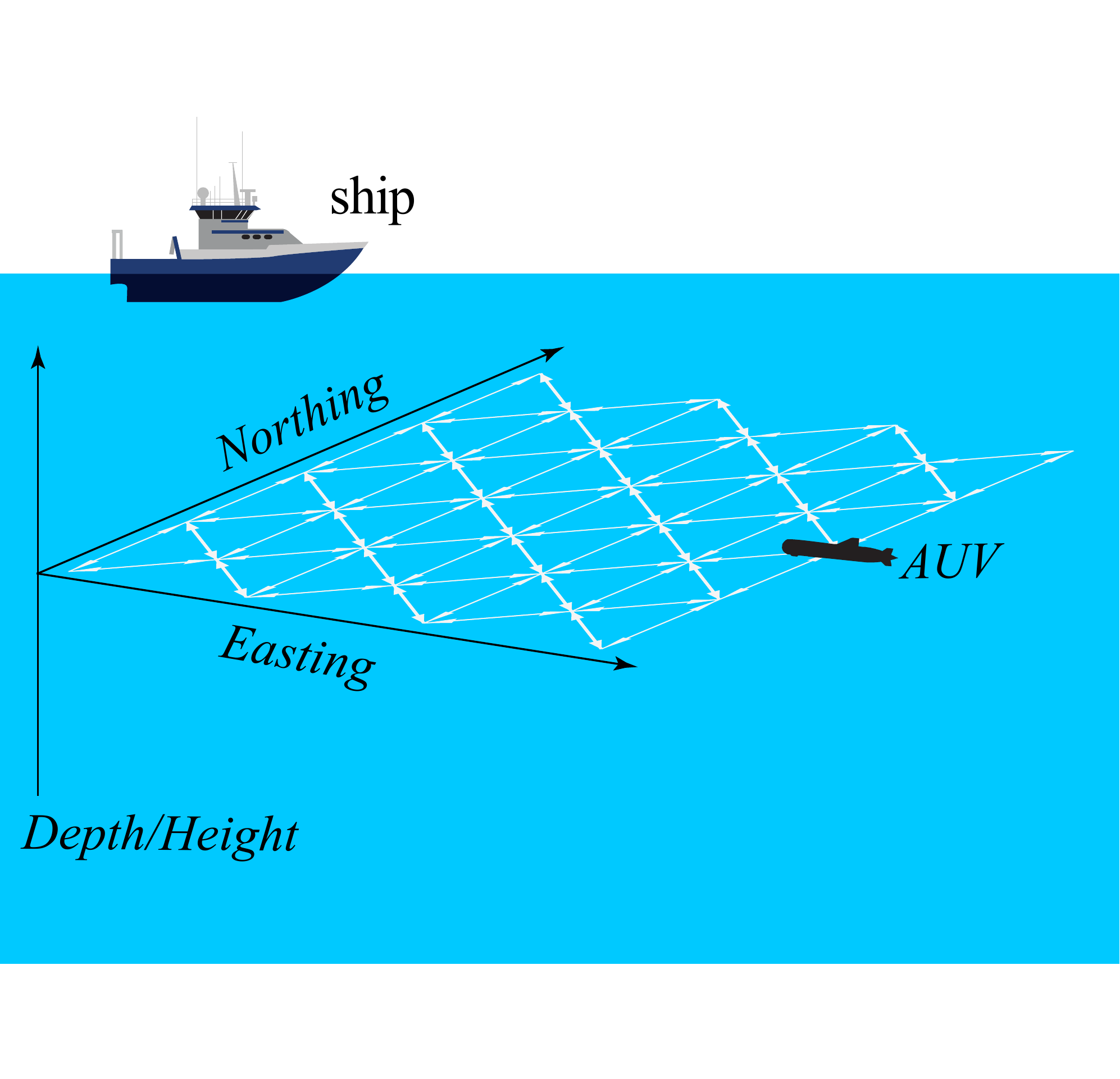}\label{fig:wp_graph_b}}
\caption{\ref{fig:wp_graph_a} The equilateral waypoint graph used to discretize the
trajectory choices over the $31\times31$ grid used to discretize the GRF. The AUV is set to start in node $53$.
\ref{fig:wp_graph_b} The waypoint grid shown in a 3D environment.}
\label{fig:wp_graph}
\end{figure}

Each strategy is conducted on an equilateral grid as shown in
Fig.~\ref{fig:wp_graph}. The AUV starts at the center coordinate at
the southern end of the domain (node 53). It then moves along edges in
the waypoint graph while collecting data which are assimilated onboard
to update the GRF model. This is used in the evaluation of the next
node to sample.  The procedure is run for $10$ stages. A total of 100
replicate simulations were conducted with all strategies.

\subsubsection{Simulation Results}

The results of the replicate runs are shown in
Fig.~\ref{fig:sim_results}, where the different criteria are plotted
as a function of survey distance. Fig.~\ref{fig:avg_ev} shows the
resulting drop in realized IBV for each of the six strategies. IBV
reduction is largest for the \textit{myopic} and
\textit{look-ahead} strategies, each performing almost equally; this
is expected as the two criteria (Eq. \eqref{critSEQ} and
\eqref{critLA}) are sensitive to differences in IBV. The \textit{static\_north} design also does well here because the path is
parallel to the boundary between the water masses.

\begin{figure}[!h]
  \centering
  \subfigure[IBV.]{\label{fig:avg_ev}\includegraphics[height=0.49\textwidth]{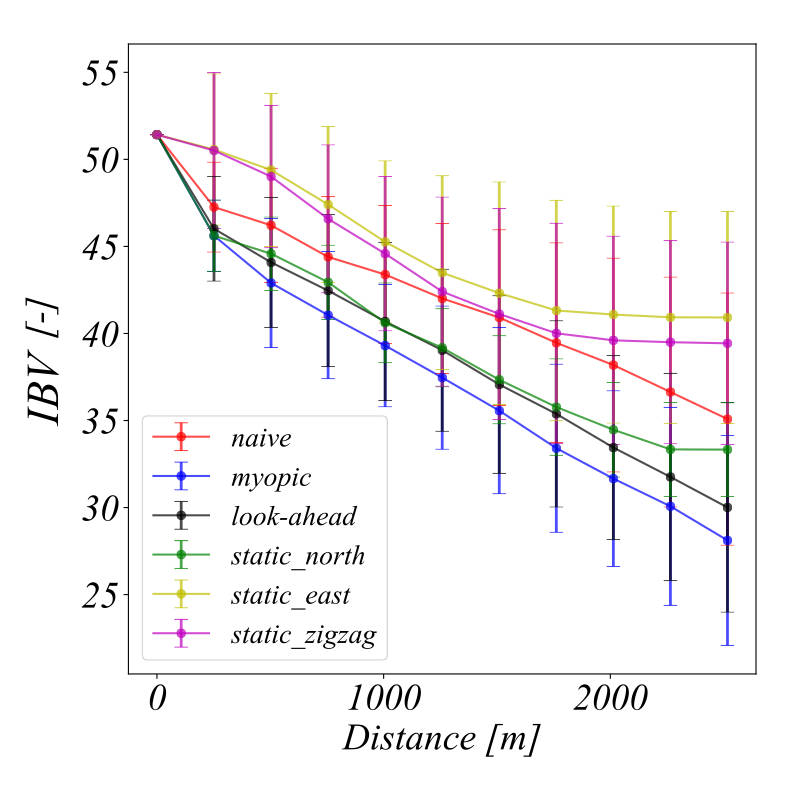}}
  \hfill \subfigure[RMSE between estimated field and
  truth.]{\label{fig:avg_rmse}\includegraphics[height=0.49\textwidth]{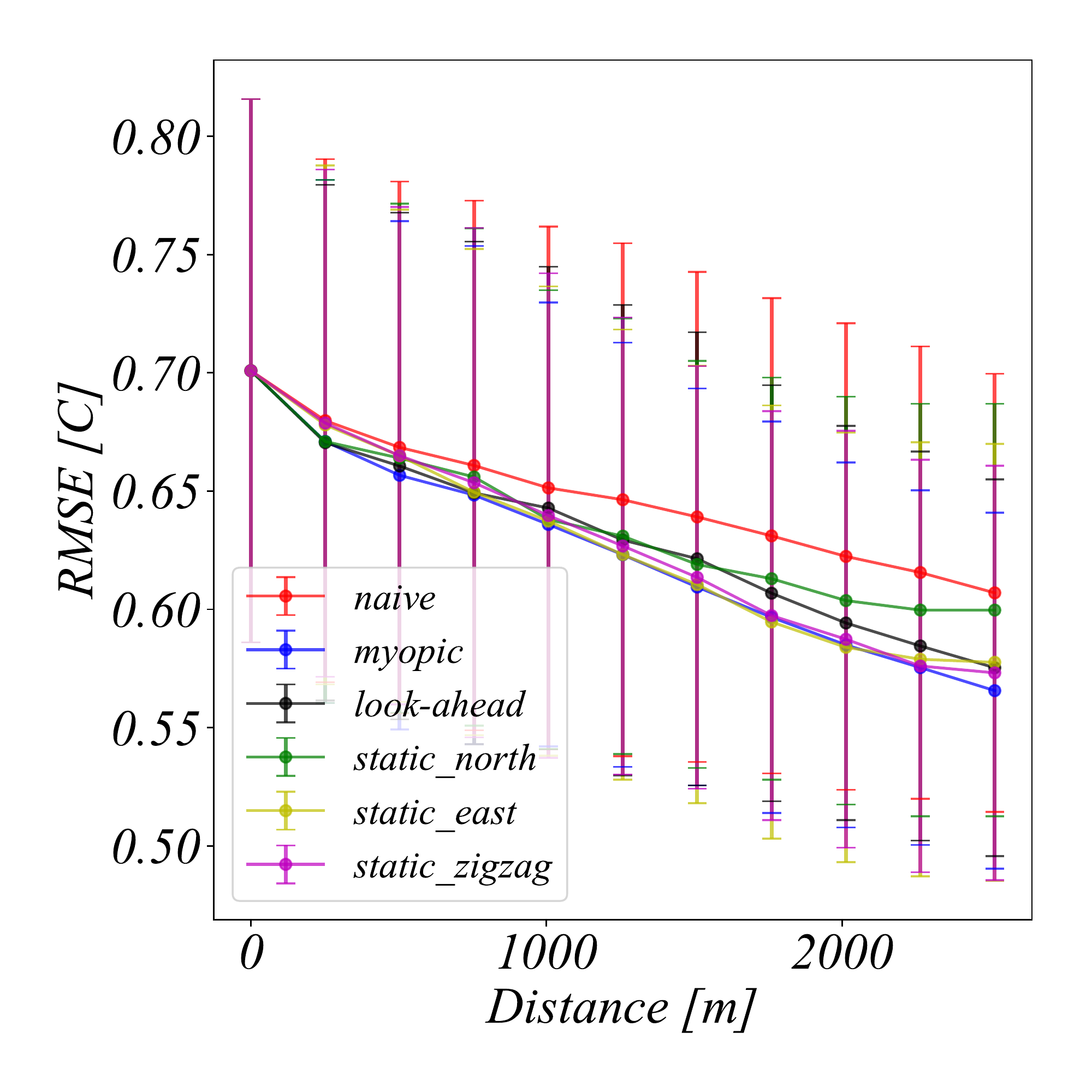}}
  \hfill \subfigure[Explained variance
  $\bR^{2}$.]{\label{fig:avg_r2}\includegraphics[height=0.49\textwidth]{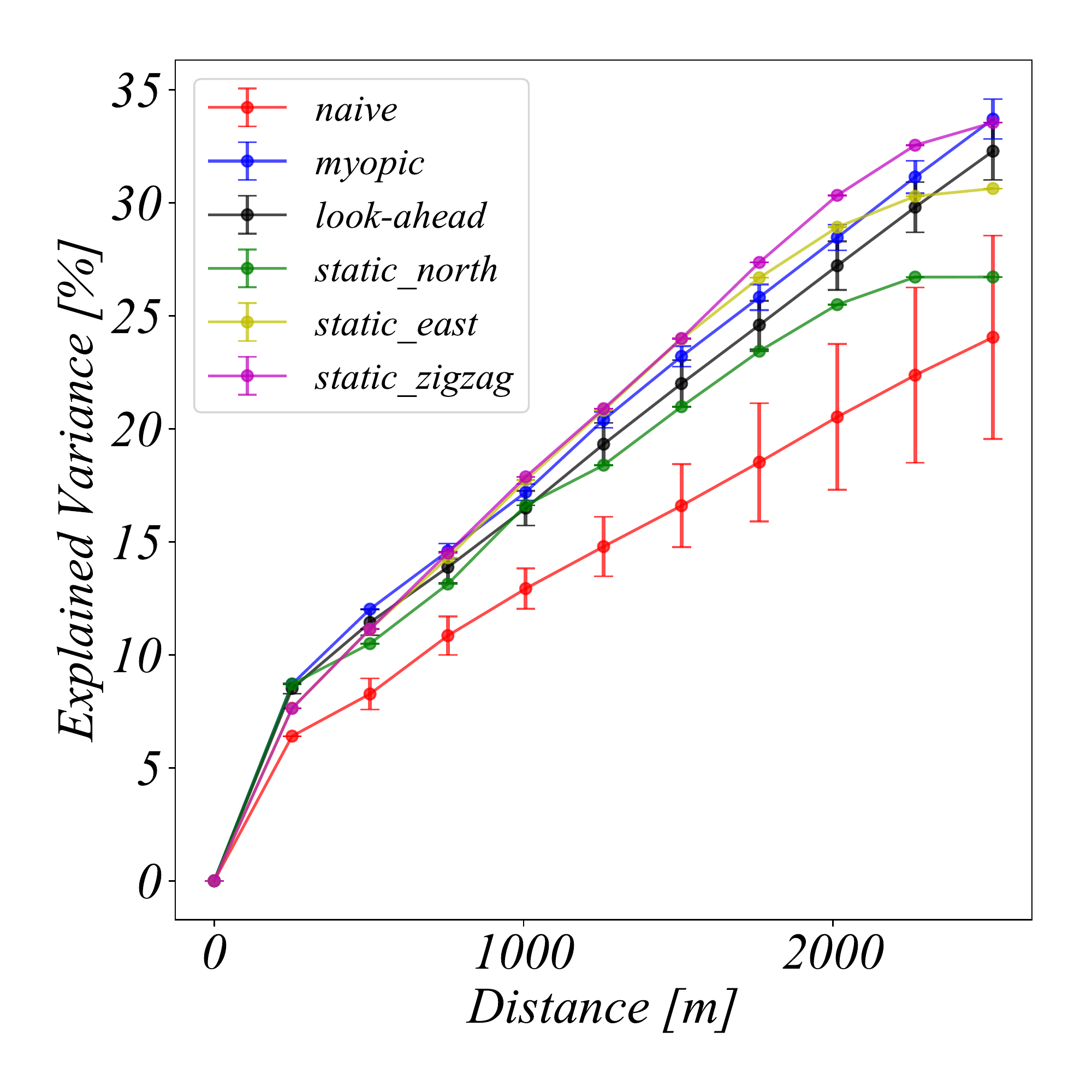}}
  \hfill \subfigure[Computational time for inference (the lines for \textit{naive}, \textit{static\_north}, \textit{static\_east}, and \textit{static\_zigzag}
  overlap).]{\label{fig:avg_time}\includegraphics[height=0.49\textwidth]{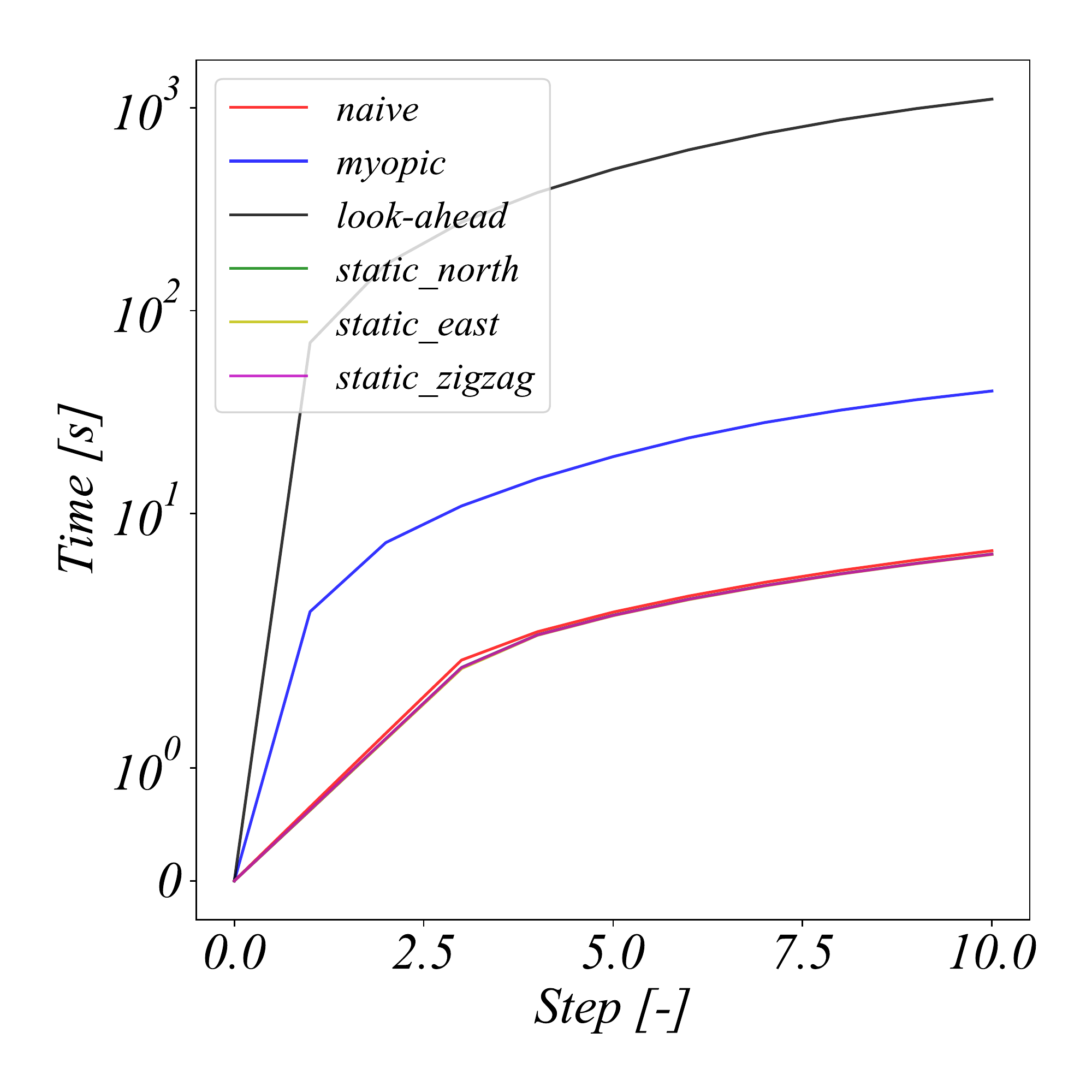}}
  \caption{Simulation results from 100 replicate simulations for 10
    sampling choices/stages on the grid. Vertical lines show variation
    in replicate results.}
\label{fig:sim_results}
\end{figure}

Fig.~\ref{fig:avg_rmse} and \ref{fig:avg_r2} show the resulting drop
in RMSE and increase in explained variance, respectively. Both
\textit{myopic} and \textit{look-ahead} strategies perform well here,
but some of the \textit{static\_east} and \textit{static\_zigzag} also
achieve good results because they cover large parts of the domain
without re-visitation. Sequential strategies targeting IBV will
sometimes not reach similar coverage, as interesting data may draw the
AUV into twists and turns. There is a relatively large variety in the
replicate results as indicated by the vertical lines. Nevertheless,
the ordering of strategies is similar.

Fig.~\ref{fig:avg_time} shows the computational effort: the
\textit{naive} strategy is on par with the static designs, while the
\textit{myopic} strategy is slower because it evaluates expected
values for all candidate directions at the waypoints. But it is still
able to do so in reasonable time, which allows for real-world
applicability. The \textit{look-ahead} strategy is much slower, reaching levels that
are nearly impractical for execution on an AUV. Some pruning of the
graph is performed to improve the performance, such as ruling out
repeated visitations. Further pruning of branches or inclusion of
other heuristics could be included for better performance. Then again,
the inclusion of such heuristics is likely a contributing factor for
the \textit{look-ahead} strategy failing to outperform the
\textit{myopic} strategy.

We studied the sensitivity of the results by modifying the input
parameters to have different correlations between temperature and
salinity, standard deviations, and spatial correlation range.  In all
runs, the \textit{myopic} and \textit{look-ahead} strategies perform
the best in terms of realized IBV, and much better than
\textit{naive}. The \textit{look-ahead} strategy seems to be
substantially better than the \textit{myopic} design only for very
small initial standard deviations or very large spatial correlation
range. 
We also ran simulation studies with only temperature data, and for
realistic correlation levels between temperature and salinity, the IBV
results are not much worse when only temperature data are
available. In addition to the comparison made in
Table~\ref{tab:sim_rhoab}, the current setting includes spatial
correlation and this likely strengthen the effect of temperature information. However, it seems that having temperature data alone
does a substantially worse job in terms of explained variance.

\section{Case Study - Mapping a River Plume}
\label{sec:case_study}

To demonstrate the applicability of using multivariate EPs and the IBV
to inform oceanographic sampling, we present a case study mapping a
river plume with an AUV. The experiment was performed in Trondheim,
Norway, surveying the Nidelva river (Fig.~\ref{fig:nidelven}). The
experiments were conducted in late Spring 2019, when there is still
snow melting in the surrounding mountains so that the river water is colder than the water in the fjord. The experiment was
focused along the frontal zone that runs more or less parallel to the
eastern shore.

\subsection{Model Specification}
\label{sec:exp_modeling}

The statistical model parameters were specified based on a short
preliminary survey where the AUV made an initial transect to determine
the trends in environmental conditions and correlation
structures. Based on the initial runs we get a reasonable idea of the temperature and salinity of river and ocean waters, and also specify the trend by linear regression, where both temperature and salinity were assumed to increase linearly with the west coordinate. Next, the residuals from the regression analysis were analyzed to specify the covariance parameters of the GRF model.
\begin{figure}[!b]
  \centering
  \subfigure[Residual plot.]{\includegraphics[width = 0.32\textwidth]{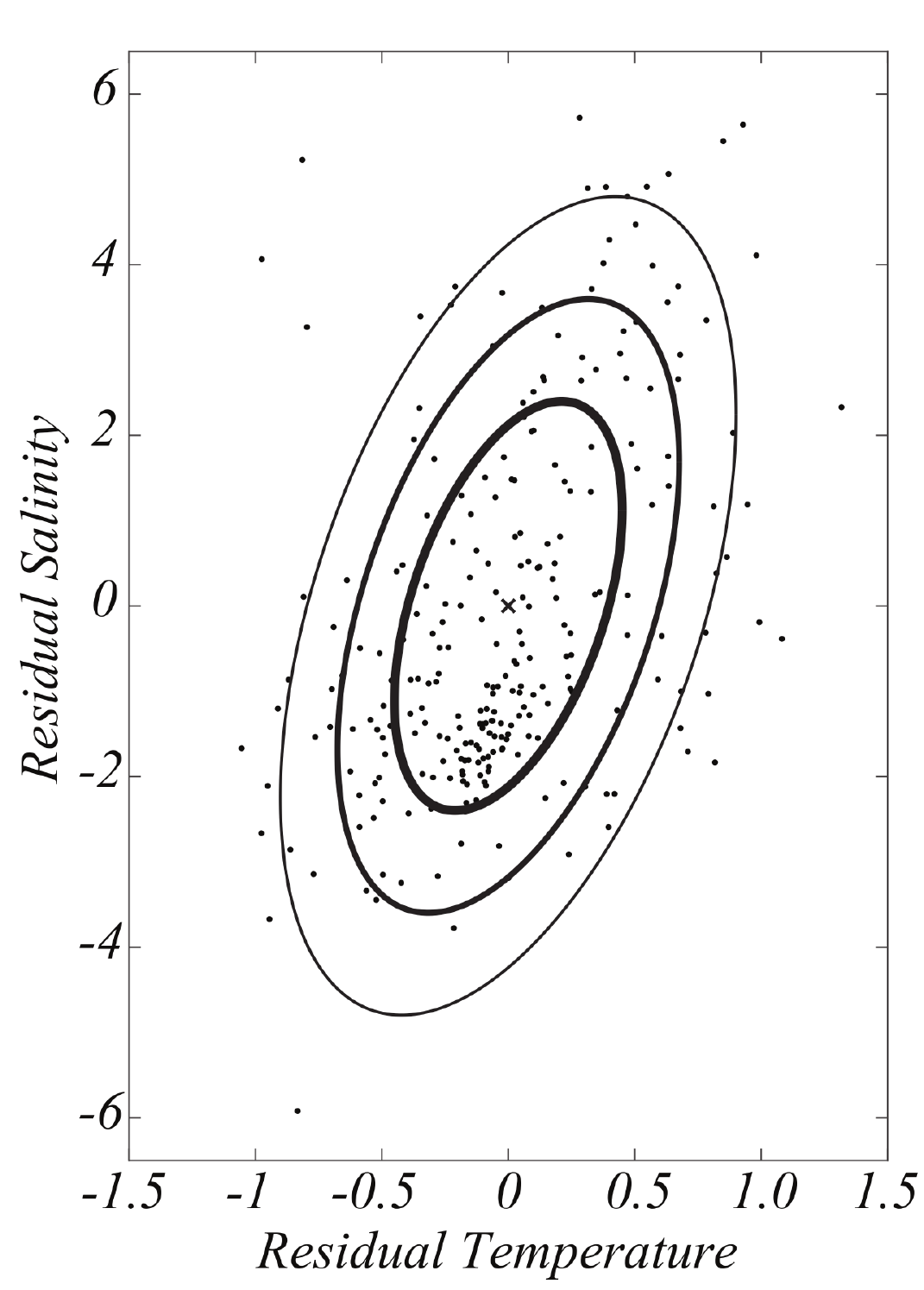}\label{fig:parest_a}}
  \hfill
  \subfigure[Empirical CDF.]{\includegraphics[width = 0.32\textwidth]{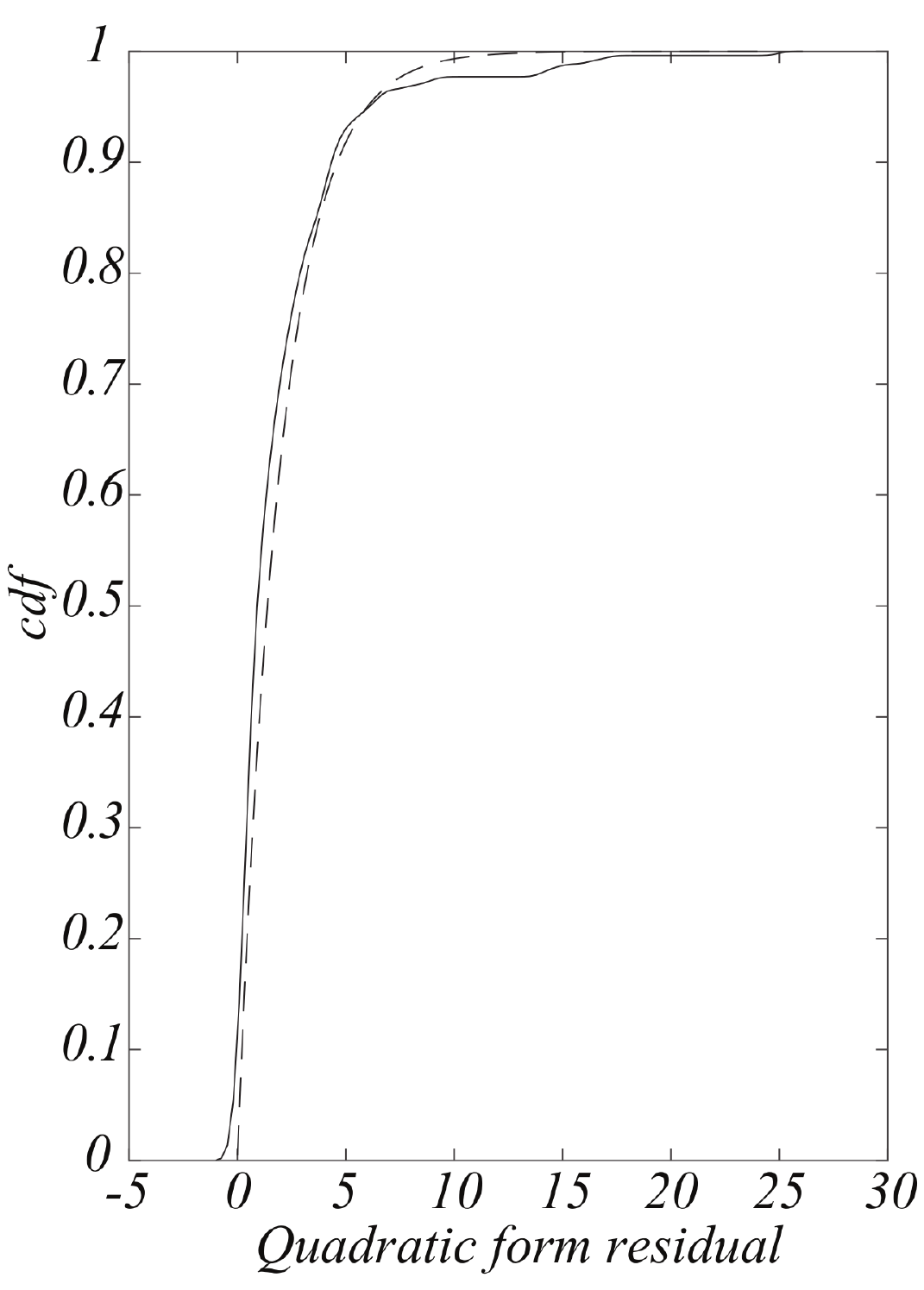}\label{fig:parest_b}}
  \hfill
  \subfigure[Empirical variogram.]{\includegraphics[width = 0.32\textwidth]{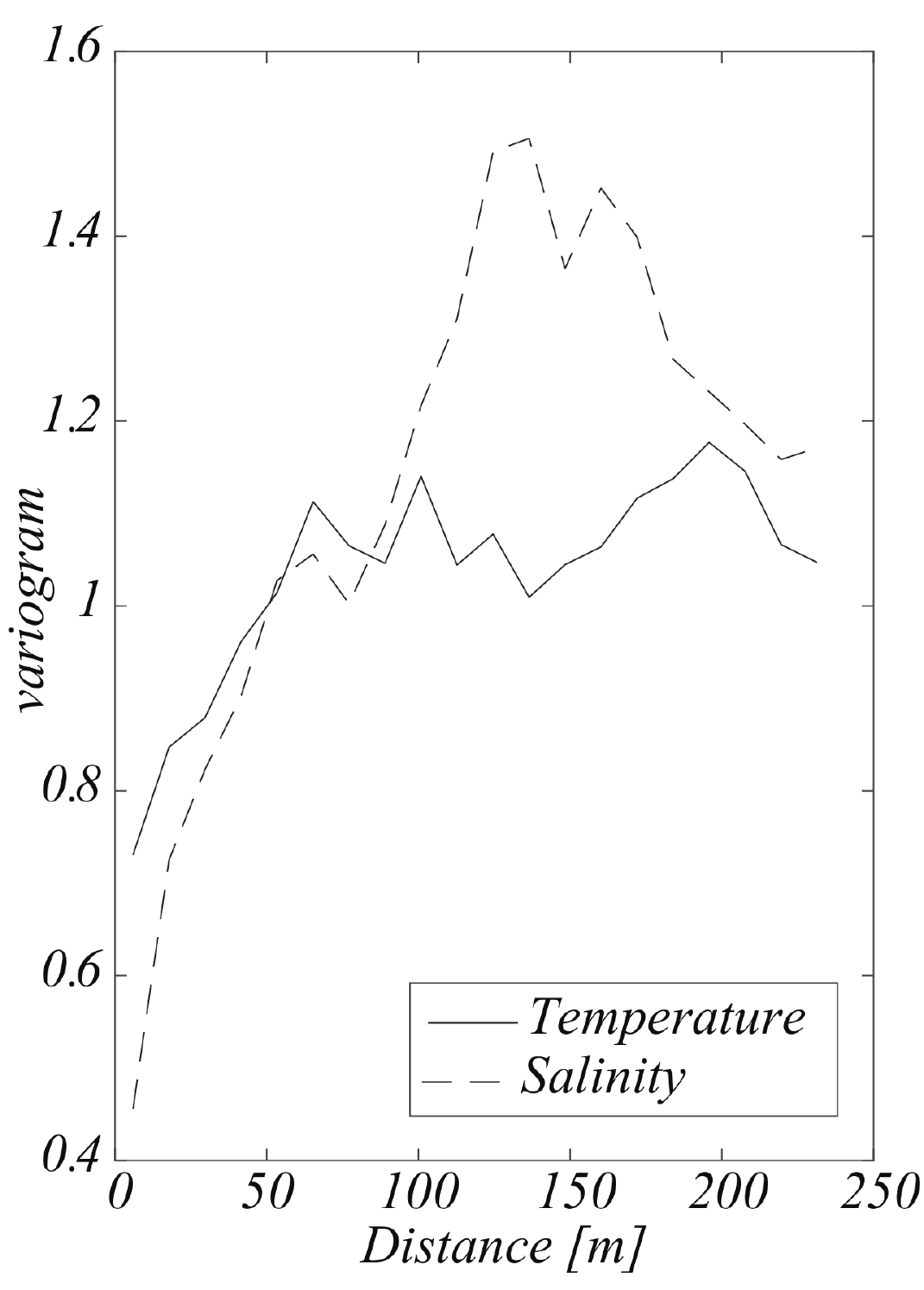}\label{fig:parest_c}}
  \caption{Data analysis from a preliminary trial experiment using the
    AUV. \ref{fig:parest_a} Residual plot of temperature and salinity
    along with Gaussian contours. \ref{fig:parest_b} Empirical CDF
    (solid) of the quadratic form of the residuals along with the
    theoretical CDF (dashed) of the $\chi^2$ distribution with two
    degrees of freedom. \ref{fig:parest_c} Empirical variogram of the
    salinity and temperature data.}
\label{fig:parest}
\end{figure}
Fig.~\ref{fig:parest} summarizes diagnostic plots of this
analysis. Fig.~\ref{fig:parest_a} shows a cross-plot of temperature
and salinity residuals after the regression mean values of salinity and
temperature are subtracted from the data. This scatter-plot of joint
residuals indicates larger variability in salinity than temperature,
and a positive correlation ($0.5$) between the two variables. Based on
the fitted bivariate covariance model (ellipses in
Fig.~\ref{fig:parest_a}), we can compute the scalar quadratic form of
the residuals, and if the model is adequate they should be approximately $\chi^2_2$ distributed. Fig.~\ref{fig:parest_b} shows
the empirical CDF of the quadratic forms (solid) together with the
theoretical CDF of the $\chi^2_2$ distribution (dashed). The modeled and
theoretical curves are rather similar, which indicates that the Gaussian
model with constant spatial variance and correlation fits reasonably well. Fig.~\ref{fig:parest_c} shows the
empirical variogram of the residuals for temperature and salinity. The decay is similar for the two, and seems to be negligible
after about $150$ m.  The working assumption of a separable covariance
function is hence not unreasonable.

Based on the analysis in Fig.~\ref{fig:parest}, the resulting
parameters are given in Table \ref{tab:experiment_param}. The
regression parameters shown here are scaled to represent the east and
west boundaries of the domain as seen in the preliminary transect
data, and the thresholds are intermediate values. These parameter
values were then used in field trials where we explored the
algorithm's ability to characterize the river plume front separating
the river and fjord water masses.

\begin{table}[!h]
\centering
\begin{tabular}{lrr}
\toprule
Parameter & Value & Source\\
\midrule
\rowcolor{Gray}
Cross correlation temperature and salinity & 0.5 & AUV observations\\
Temperature variance &  0.20 & AUV observations (variogram)\\
\rowcolor{Gray}
Salinity variance &  5.76 & AUV observations (variogram)\\
Correlation range  & 0.15 km & AUV observations (variogram)\\
\rowcolor{Gray}
River temperature  & $10.0\,^{\circ}\mathrm{C}$ & AUV observations\\
Ocean temperature $T_{ocean}$ & $11.0\,^{\circ}\mathrm{C}$ & AUV observations\\
\rowcolor{Gray}
River salinity $S_{river}$ & $14.0$ g/kg & AUV observations\\
Ocean salinity $S_{ocean}$ & $22.0$ g/kg & AUV observations\\
\rowcolor{Gray}
Threshold in temperature & $10.5\,^{\circ}\mathrm{C}$ & User specified \\
Threshold in salinity & $18.0$ g/kg & User specified \\
\rowcolor{Gray}
\bottomrule
\end{tabular}
\caption{Model and threshold parameters from an initial AUV
  survey. Observations were taken across the front while crossing from
  fresh, cold river water to saline and warmer ocean waters.}
\label{tab:experiment_param}
\end{table}

\subsection{Experimental Setup}

A Light AUV \citep{sousa2012lauv} (Fig.~\ref{fig:lauv}) equipped with
a 16 Hz Seabird Fastcat-49 conductivity, temperature, and depth (CTD)
sensor was used to provide salinity and temperature measurements.  The
AUV is a powered untethered platform that operates at $1$-$3$ m/s in
the upper water column. It has a multicore GPU NVIDIA Jetson TX1
(quad-core 1.91 GHz 64-bit ARM machine, a 2-MB L2 shared cache, and 4
GB of 1600 MHz DRAM) for computation onboard.  The sampling algorithm
was built on top of the autonomous Teleo-Reactive EXecutive
(\textit{T-REX}) framework \citep{py10,Rajan12,Rajan12b}. We assume
that the measurements are conditionally independent because the
salinity is extracted from the conductivity sensor which is different
from the temperature sensor. We specify variance $0.25^2$ for both
errors, which is based on a middle ground between the nugget effect in
the empirical variogram and the sensor specifications.

\begin{figure}[!h] 
\centering 
\includegraphics[width=0.98\textwidth]{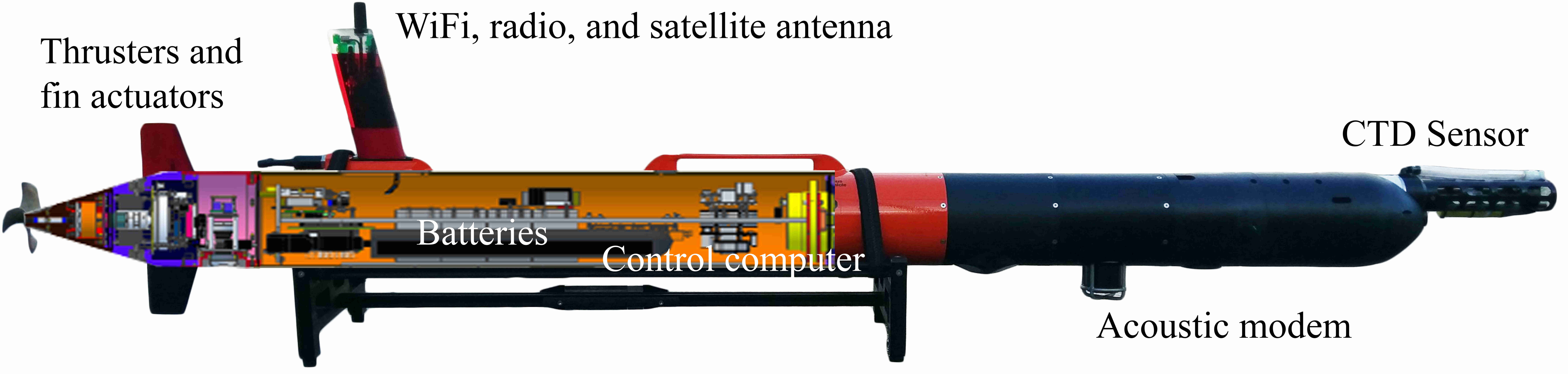}
\caption{The commercially available Light Autonomous Underwater
  Vehicle (LAUV) platform for upper water-column exploration used in
  our experiments.}
\label{fig:lauv}
\end{figure} 

The AUV was running a \textit{myopic} strategy to decide between
sampling locations on the waypoint graph distributed over an
equilateral grid, as shown in the grey-colored lattice in
Fig.~\ref{fig:map}.  At each stage, it takes the AUV about 30 seconds
to assimilate data and evaluate the EIBV for all the possible waypoint alternatives.  It
was set to start in the south-center part of the waypoint graph. A
survey was set to take approximately 40 minutes, visiting 15 waypoints
on the grid, with the vehicle running near the surface to capture the
plume. On its path from one waypoint to the next, the AUV collects
data with an update frequency of 30 seconds, giving three measurements
per batch in the updates at each stage.

\begin{figure*}[!h]
\centering
\subfigure[AUV survey area]{\includegraphics[height=0.41\textwidth]{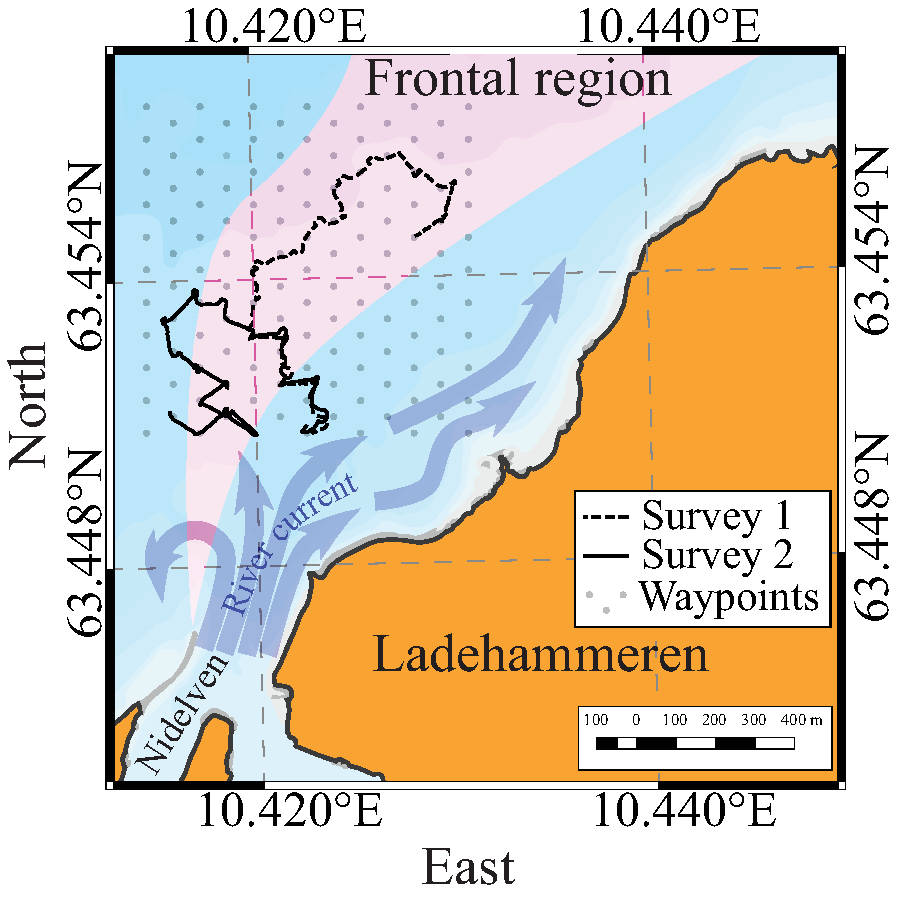}\label{fig:map}}
\hspace{0.3cm}
\subfigure[Temperature tracks]{\includegraphics[height=0.41\textwidth]{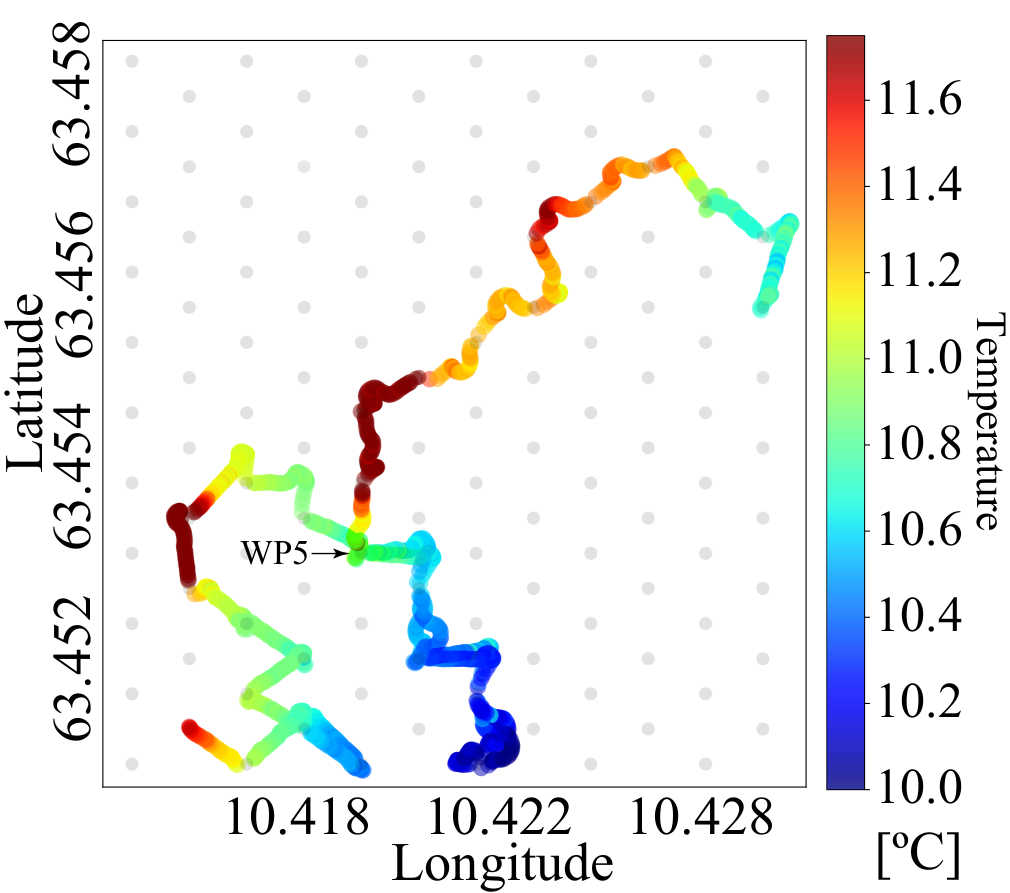}\label{fig:res_both}}

\subfigure[Survey 1]{\includegraphics[height=0.40\textwidth]{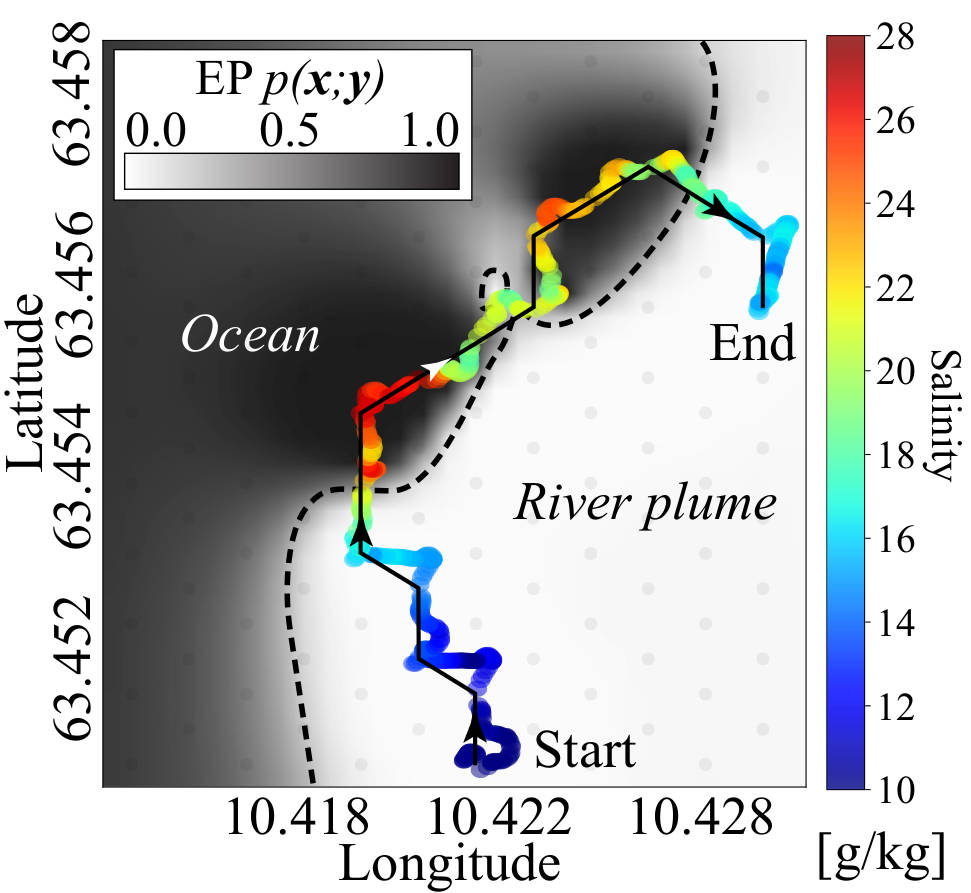}\label{fig:res1}}
\hspace{0.2cm}
\subfigure[Survey 2]{\includegraphics[height=0.40\textwidth]{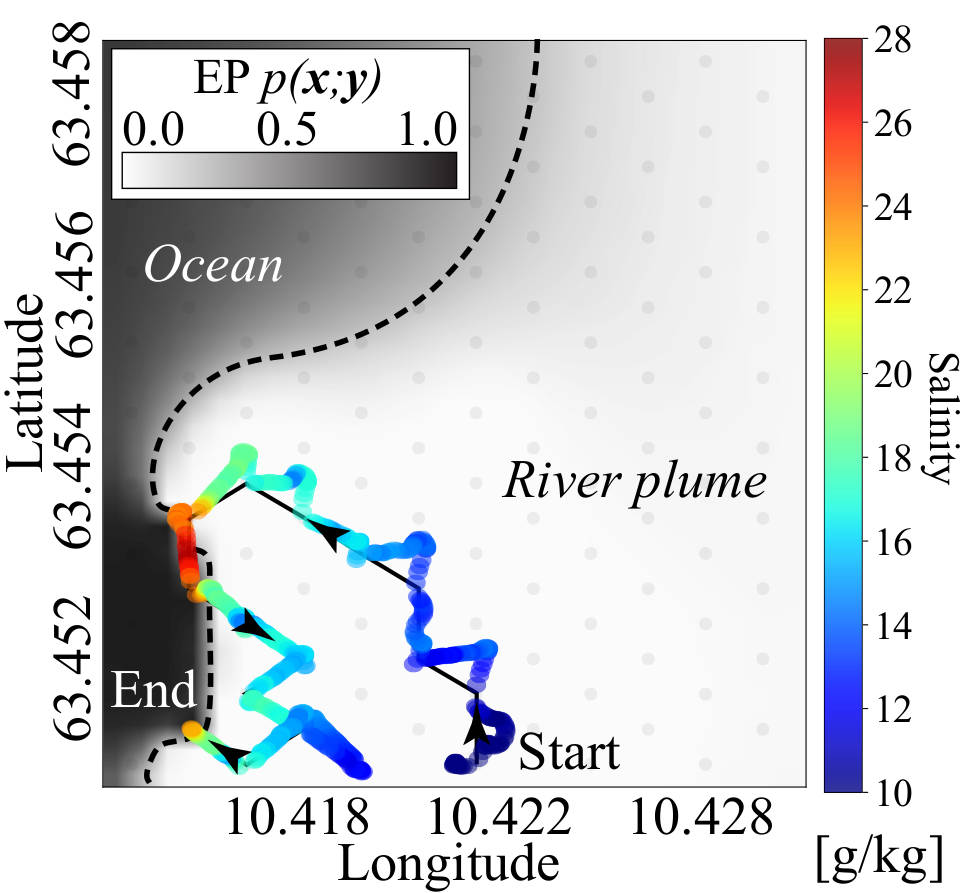}\label{fig:res2}}
\caption{Results from mapping the Nidelva river, Trondheim, Norway
  over two survey missions. \ref{fig:map} shows an overview of the
  survey area overlaid with the AUV path in black and dashed
  line. Note the shaded region indicating a typical frontal
  region. \ref{fig:res_both} shows the collected temperature data as
  colored trails. Note waypoint 5 (WP5) which indicates where the two
  surveys diverge. \ref{fig:res1} and \ref{fig:res2} shows the
  collected salinity data overlaid on the final EP, which indicate the
  AUVs statistical impression of the front. For both missions the
  temperature and salinity data correspond with an indication of the
  EP front. About 2 hours time separated the two runs.}
\label{fig:results}
\end{figure*}

\subsection{Results}

Two survey missions (1 and 2), were run successively, with a short
break in between. The resulting path of the selected waypoints are
shown in the map in Fig.~\ref{fig:map}, both within the expected
frontal region (shaded pink). The recorded temperatures are shown as
colored trails in Fig.~\ref{fig:res_both}, clearly indicating the
temperature difference between fjord and riverine waters. The salinity
data are then shown separately, overlaid with the estimated EP for
each survey in Fig.~\ref{fig:res1} and Fig.~\ref{fig:res2}.

Both surveys successfully estimated and navigated the separation zone,
crossing the frontal boundary multiple times. As conditions changed
slightly between the two surveys, the resulting trajectory (after
waypoint 5) is shown to deviate. Survey 1 continued northwards,
tracking the north-eastern portion of the front, while Survey 2 turned
west, mapping the south-western region.

The final predictions of the front location, represented by
conditional EPs in Fig.~\ref{fig:res1} and Fig.~\ref{fig:res2} as
dashed lines, correspond with one another. In both surveys they yield
a picture of the front being to the west in the southern portions of
the region and gradually bending off toward the north east. The amount
of exploration done by Survey 1 which turned north is greater than
Survey 2 which was coming close to the survey area borders in the
south-western corner.

\section{Closing remarks}
\label{sec:concl_disc}

This work builds on a multidisciplinary effort combining statistical
methods with robotic sampling for oceanographic applications. We show
how observation practices can gain efficiency and accuracy from
statistical techniques for spatial monitoring and demonstrate the the
need for real-time multivariate spatial sampling on autonomous
platforms.

In particular, we derive and show results for a real-world domain
characterizing water mass properties. The characterization of
uncertainties in random sets is extended in the vector-valued case
with new results for the expected integrated Bernoulli variance
reduction achieved by spatial sampling designs. This is provided in
semi-analytical form for static designs, and then extended to the
adaptive situations. The sequential derivations provide new insights
into efficient applications of adaptive data collection, as
demonstrated in our application.

The case study considers the upper water column in the river plume,
represented by a two dimensional grid. Extensions to three-dimensional
domains are not methodologically different, but likely approximate
calculations by concentrating numerical integration on terms in the
vicinity of the autonomous vehicle \citep{fossum18b}. While we did not
consider any temporal effects, which would be relevant on a larger
time scale, we do consider the extension to spatio-temporal modeling
and envision that advection-diffusion equations could be useful
\citep{sigrist2015stochastic,richardson2017sparsity}. For more complex
oceanographic phenomena, the methods will need to be extended to
non-Gaussian phenomena, possibly feature-based mixtures of Gaussian
processes which could potentially be run onboard augmented by
dynamical models. Running numerical models onboard a robotic vehicle
is currently infeasible, but high-resolution ocean models or remote
sensing data can be used to fit a more complex statistical model
\cite{davidson19}.

The spatio-statistical design criterion building on random sets is
relevant in our setting with different water properties.  We show
mathematical generality beyond the expected integrated Bernoulli
variance, for instance, that of volume uncertainties which is possibly
more relevant, but one that requires more computational resources.
Such criteria could be particularly useful in other oceanographic
settings related to mapping of algal-blooms, anoxic zones or open
water fronts \cite{costa19}.  Other criteria could also be relevant,
for instance, hybrid or multi-attribute criteria that could balance
goals of exploration and exploitation in this situation. Equally, such
techniques have significant use cases in downstream decision-making,
with policy makers and regulators who need to make difficult decisions
related to aquaculture or other marine resources. Value of information
analysis \citep{Eidsvik:15} could be used to evaluate whether
information is likely to result in improved decision-making, in such a
context.  We also foresee opportunities related to design of
experiments for multivariate processes using our notion of generalized
locations.

In our context the myopic strategy performs well, and due to
computational constraints we did not go in depth on the dynamic
programming solutions. There has been much work on finite horizon
optimization in the robotics literature including probabilistic road
maps and rapidly-exploring random trees \citep{karaman2011sampling},
but their statistical properties are unclear.  In some cases it is
also limiting to use a waypoint graph, and would be beneficial to
allow more continuous updates and navigation at the highest frequency
possible given limitations being onboard an AUV.  It is equally
interesting to explore the additional flexibility that can be gained
by having multiple vehicles co-temporally exploring a spatial or
spatio-temporal domain \citep{ferreira2019advancing}. Such an approach
would enable concurrent sampling in different parts of the space, or
opportunities to move in parallel to best capture the excursion set.
The value of information related to when and what to communicate (to
shore or to other vehicles) is also an interesting thrust for research
and likely to be useful for internet-of-things applications or
computer experiments where some observations or evaluations are rather
inexpensive, while others must only be done when they are really
valuable.

\section*{Acknowledgements}

TOF acknowledges support from the Centre for Autonomous Marine
Operations and Systems (AMOS), Center of Excellence, project number
223254, and the Applied Underwater Robotics Labortatory (AURLab). CT
and DG acknowledge support from the Swiss National Science Foundation,
project number 178858. JE and KR acknowledge support from Norwegian
research council (RCN), project number 305445. DG would like to
aknowledge support of Idiap Research Institute, his primary
affiliation in an early version of this manuscript. The authors would
also like to thank Niklas Linde of the University of Lausanne for
providing constructive feedback about this work, and members of the
NTNU AURLab for help with AUV deployments

\footnotesize
\bibliographystyle{imsart-nameyear}
\bibliography{ref}

\section*{Appendix}

\begin{propo}
    \label{propo1}
For a measurable random field $\gp$ and a locally finite measure $\mes$ on $\domain$, $\mes(\es)$ is a random variable and for 
any $r\geq 1$,
\begin{equation*}
\begin{split}
\mathbb{E}[\mes(\es)^r]
&=\int_{\domain^{r}} \jointExcuProb
\productMeasure
,
\end{split}
\end{equation*}

where the product measure is denoted as
$\nu^{\otimes}:=\bigotimes_{i=1}^r \nu$.
Here $\gp$ is defined on $\domain$, and for
$\bm{u}=\left(u^{(1)}, ..., u^{(r)}\right)\in \domain^r$, $\gp[\bm{u}]=\left(\gp[u^{(1)}], ...,
\gp[u^{(r)}]\right)\in \mathbb{R}^{\no r}$.
\medskip

In the particular case where $\gp$ is a multivariate Gaussian random field
we have
\begin{align*}
\jointExcuProb = \mathcal{N}_{\no r}(T^r; \meanUU, \covUU),
\end{align*}
where $\mathcal{N}_{\no r}(\cdot ; \meanUU, \covUU)$ is the Gaussian measure on $\mathbb{R}^{\no r}$ with mean $\meanUU$ 
and covariance matrix $\covUU$, respectively defined blockwise by
\begin{align*}
\meanUU&=\begin{pmatrix}\mu(u^{(1)})\\ \vdots\\ \mu(u^{(r)})\end{pmatrix}
\in \mathbb{R}^{\no r}, \\
\text{and } \covUU &= \begin{pmatrix}
\cov(\gp[u^{(1)}], \gp[u^{(1)}]) & \dots & \cov(\gp[u^{(1)}],
\gp[u^{(r)}])\\
\vdots & & \vdots\\
\cov(\gp[u^{(r)}], \gp[u^{(1)}]) & \dots & \cov(\gp[u^{(r)}],
\gp[u^{(r)}])\\
\end{pmatrix}\in \mathbb{R}^{pr\times pr},
\end{align*}
each of the $r\times r$ blocks of the latter matrix being itself a (cross-)covariance matrix of dimension $\no \times 
\no$.
Assuming further that $\covUU$ is non-singular, the probability of interest can be formulated in terms of the $\no 
r$-dimensional Gaussian probability density function
$\varphi_{\no  r}(\cdot;~\meanUU, \covUU)$ as
\begin{equation*}
\begin{split}
&\jointExcuProb
=
\int_{T^r} \varphi_{\no  r}\left(\bm{v};
    ~\meanUU, \covUU\right)
    \mathrm{d}\bm{v},
\end{split}
\end{equation*}
In the particular orthant case with $\T=(-\infty, t_1] \times \dots \times (-\infty, t_{r}]$,
the latter probability directly writes in terms of the multivariate Gaussian
cumulative distribution, 
this time by the way without requiring $\covUU$ 
to be non-singular:
\begin{equation*}
\begin{split}
\jointExcuProb
&=
\varPhi_{\no r}\left(\bm{t};~\meanUU, \covUU\right),
\end{split}
\end{equation*}
where we have used the notations
$t=(t_1,\dots,t_{\no}
)\in\mathbb{R}^{\no}$, $1_{r}=(1,\dots,1)\in \R^{r}$, and
$\bm{t}=1_{r}\otimes \bm{t}=(t_1,\dots,t_{\no},\dots,t_1,\dots,t_{\no})
\in \R^{\no r}$.
\end{propo}
\begin{proof}
That $\mes(\es)$ defines indeed a random variable follows from Fubini's theorem
relying on the joint measurability of
$(\x, \omega) \to \mathbbm{1}_{\es(\omega)}(\x)$,
itself inherited from the assumed measurability for
$(\x, \omega) \to \gp[\x](\omega)$ and $T$, respectively. From there, following the steps of Robbins' theorem \cite{Robbins1944}, we find that
\begin{equation*}
\begin{split}
\mathbb{E}[\mes(\es)^{r}]
&=\mathbb{E}\left[\left(\int_{\domain} \mathbbm{1}_{\gp[u] \in T} ~d\mes(u) \right)^{r} \right]
=\mathbb{E}\left[ \prod_{i=1}^{r} \left(
        \int_{\domain} \mathbbm{1}_{\gp[\uu^{(i)}] \in T} ~d\mes(\uu^{(i)})
\right) \right] \\
&
=
\mathbb{E}\left[
\int_{\domain^{r}}
\mathbbm{1}_{\gp[\uu^{(1)}] \in T,\dots, \gp[\uu^{(r)}]  \in T}
\productMeasure
\right]
=\int_{\domain^{r}}
\jointExcuProb
\productMeasure.
\end{split}
\end{equation*}
The rest consists in expliciting the probability of $T\times \dots \times T$ under the multivariate Gaussian distribution of
$\left(\gp[u^{(1)}], \dots,  \gp[u^{(r)}] \right)$.
\end{proof}

The propositions below provide formulae for computations of expectations of moments of multivariate gaussian CDFs.

\begin{propo}
    \label{propo2}
Let $p, q, h \geq 1$, $a \in \R^p$, $B \in \R^{p\times q}$,
and $\covN$, $\covV$ be two covariance matrices in
$\R^{p\times p}$ and $\R^{q\times q}$, respectively.
Then, for $V \sim \mathcal{N}_{q}(0_q, \covV)$,
\begin{equation*}
\mathbb{E}\left[ \varPhi_{p}\left( a + BV; \covN \right)^h \right]
=
\varPhi_{ph}
\left(
    \bm{a}
;~
\bm{\Sigma}
\right),
\end{equation*}
where the vector $\bm{a} \in \R^{p h}$ is defined as
$\bm{a} := 1_h\otimes a = 
\left(a, \dots , a
\right)'$
 and the $p h\times p h$ covariance matrix is given by
 $\bm{\Sigma} := 
1_h 1_h'\otimes B\covV B' + I_h\otimes \covN$.
\end{propo}

\begin{remark}
In blockwise representation, $\bm{\Sigma}$ can be expressed as follows:
\begin{align*}
\begin{pmatrix}
    \covN & &\\
        & \ddots &\\
        &   & \covN
\end{pmatrix}
+
\begin{pmatrix}
B\covV B' & \dots & B\covV  B'\\
\vdots & & \vdots\\
B\covV B' & \dots & B\covV B'\\
\end{pmatrix}
\end{align*}
\end{remark}

\begin{proof}
By definition of $\Phi_{p}$, for $N\sim \mathcal{N}_{p}(0_{p},\covN)$,
$$
\mathbb{P}(N\leq a + BV | V)
=
\varPhi_{p}\left( a + BV; \covN \right).
$$
Now for $\varPhi_{p}\left( a + BV; \covN \right)^h$, provided that the probability space is sufficiently large to accomodate $h$ independent Gaussian random vectors $N_i\sim \mathcal{N}_{p}(0,\covN)$ (which is silently assumed here), using the former equality delivers
$$
\varPhi_{p}\left( a + BV; \covN \right)^h
=
\prod_{i=1}^h \mathbb{P}(N_i\leq a + BV | V).
$$
Now by independence of the $N_i$'s we obtain the joint conditional probability
$$
\prod_{i=1}^h \mathbb{P}(N_i\leq a + BV | V)
=
\mathbb{P}(N_1\leq a + BV, \dots, N_h\leq a + BV| V),
$$
whereof, by virtue of the law of total expectation,
\begin{equation*}
\begin{split}
\mathbb{E}\left[ \varPhi_{p}\left( a + BV; \covN \right)^h \right]
&=\mathbb{E}\left[\mathbb{P}(N_1\leq a + BV, \dots, N_h\leq a + BV| V)\right]\\
&=\mathbb{P}(N_1\leq a + BV, \dots, N_h\leq a + BV)\\
&=\mathbb{P}(W_1 \leq a, \dots, W_h\leq a)\\
&=\varPhi_{ph}
\left(
1_{h} \otimes a
;
(1_{h}1_{h}')\otimes (B\Sigma_{V} B') + 
I_{h}\otimes \covN
\right),
\end{split}
\end{equation*}
where $\mathbf{W}=(W_1,\dots,W_h)$ with $W_i=N_i- BV \ (1\leq i \leq h)$
and the last line follows $\mathbf{W}$ forming a Gaussian vector (by global independence of the $N_i$'s and $V$) and from the definition of $\varPhi_{p h}$. The covariance matrix $\mathbf{\Sigma}$ of $\mathbf{W}$ is obtained by noting that $\operatorname{cov}(W_i,W_j)=B \covV B' + \delta_{ij} \covN \ 
(i,j \in \{1,\dots,h\})$. 
\end{proof}

We now generalize Proposition~\ref{propo2} to the case of multivariate monomials in orthant probabilities with thresholds affine in a common Gaussian vector.

\begin{propo}
    \label{propo3}
Let $g, p, q\geq 1$, $h_{1},\dots, h_{g}\geq 1$ with $H=\sum_{i=1}^g h_i$, $a_{i} \in \R^{p}$, $B_{i}\in \R^{p \times q}$, and covariance matrices $\covN_i \in \R^{p \times p}$ $(1\leq i \leq g)$. Then, for any covariance matrix $\covV \in \R^{q\times q}$ and $V\sim\mathcal{N}_{q}(0_q,\covV)$,
    \begin{equation}
    \mathbb{E}\left[ \prod_{i=1}^{g} \varPhi_{p}\left(a_i + B_{i}V; \covN_{i} \right)^{h_i} \right]
    =
\varPhi_{p H}
\left(
    \bm{a}
;
\mathbf{\Sigma}
\right),
\end{equation}
with $\bm{a}=(1_{h_1}\otimes a_1, \dots, 1_{h_g}\otimes a_{g}) \in \R^{p H}$
and $\mathbf{\Sigma}\in \R^{p H \times p H}$ is defined blockwise by $(\Sigma_{i,j})_{i,j \in \{1,\dots, g\}}$ where, for any $i,j \in \{1,\dots, g\}$, 
\begin{equation}
\Sigma_{i,j}=
(1_{h_{i}}1_{h_{j}}')\otimes (B_{i}\Sigma_{V} B_{j}') + \delta_{i,j}(I_{h_{i}}\otimes \covN_{i}) \in \R^{p h_{i} \times p h_{j}}.
\end{equation}
\end{propo}

\begin{remark}
Using blockwise representation for the blocks themselves delivers 
\begin{equation*}
\Sigma_{ij} =
\begin{pmatrix}
B_i\Sigma_{V} B_j' & \dots & B_i\Sigma_{V} B_j'\\
\vdots & & \vdots\\
B_i\Sigma_{V} B_j' & \dots & B_i\Sigma_{V} B_j'\\
\end{pmatrix}
+
\delta_{ij}
\begin{pmatrix}
    \covN_i & &\\
        & \ddots &\\
        &   & \covN_i
\end{pmatrix}
\end{equation*}
Here each $\Sigma_{ij}$ is made of $h_i$ times $h_j$
(vertically/horizontally) $p \times p$ sub-blocks, hence possesses $ph_i$ lines and $ph_j$ columns.
\end{remark}

\begin{proof}
    The proof relies (again) heavily on the fact that, by definition of $\Phi_{p}$, for any covariance matrix $\covN \in \R^{p \times p}$, $a\in \R^p$, $B\in \R^{p \times q}$, and $N\sim \mathcal{N}_{p}(0_{p},\covN)$,
    $$
    \mathbb{P}(N\leq a + BV | V)
    =
    \varPhi_{p}\left( a + BV; \covN \right).
    $$
In particular, for globally independent $N_{i,j} \sim \mathcal{N}_{p}(0_{p},\covN_i)$ $(1\leq j \leq h_i, 1\leq i \leq g)$,
\begin{equation*}
\begin{split}
\prod_{i=1}^{g} \varPhi_{p}\left(a_i + B_{i}V; \covN_{i} \right)^{h_i}
&=
\prod_{i=1}^{g}
\prod_{j=1}^{h_{i}}
\mathbb{P}(N_{i,j}\leq a_i + B_i V | V)\\
&=\mathbb{P}(N_{1,1}\leq a_1 + B_1 V, \dots, N_{g,h_{g}}\leq a_g + B_g V | V),
\end{split} 
\end{equation*}
so that, by the law of total expectation,
    \begin{equation*}
    \begin{split}
    \mathbb{E}\left[ \prod_{i=1}^{g} \varPhi_{p}\left(a_i + B_{i}V; \covN_{i} \right)^{h_i} \right]
=
\mathbb{P}(W_{1} \leq 1_{h_{1}} \otimes a_1, \dots, W_{g} \leq 1_{h_{g}} \otimes a_g)
    \end{split}
    \end{equation*}
where
$W_{1}=(N_{1,1}- B_1 V, \dots, N_{1,h_{1}}- B_1 V), 
W_{2}=(N_{2,1}- B_2 V, \dots, N_{2,h_{2}}- B_2 V), 
\dots, W_{g}=(N_{g,1}- B_g V, \dots, N_{g,h_{g}}- B_g V)$. Noting that $\mathbf{W}=(W_1,\dots, W_g)$ is a centred $p H$-dimensional Gaussian random vector, we finally obtain that
    \begin{equation*}
\begin{split}
\mathbb{E}\left[ \prod_{i=1}^{g} \varPhi_{p}\left(a_i + B_{i}V; \covN_{i} \right)^{h_i} \right]
=
\varPhi_{p H}\left(\bm{a};\mathbf{\Sigma}\right),
    \end{split}
\end{equation*}
with $\bm{a}=(1_{h_{1}} \otimes a_1, \dots, 1_{h_{g}} \otimes a_g)$ and $\bm{\Sigma}=(\operatorname{cov}(W_i,W_j))_{i,j \in \{1,\dots, g\}}$.
\end{proof}

Those two general results allow us to derive simple expressions for the expected effect of the inclusion of new datapoints on the $\IBV$ (Proposition \ref{propo_eibv}) and on the $\EMV$ (Proposition \ref{propo_emv}) for which we provide proofs below.

\begin{proof}{(Proposition \ref{propo_eibv})}
Applying Tonelli-Fubini followed by the law of total
expectation first delivers
\begin{equation*}
\begin{split}
\eibv_{[n]}(\bm{x})
&=\int_{\domain}
\currentExp{\futureProba{\gp[\uu]\in
        T}(1-\futureProba{\gp[\uu]\in T})} d\mes(u) \\
&=\int_{\domain} \varPhi_{\no}\left(\bt;
~\futureMean{\uu},
\futureCov{u,u}\right) d\mes(u)\\
&-\int_{\domain} \currentExp{
    \varPhi_{\no}\left(\bt;
    ~\futureMean{\uu},
    \futureCov{u,u}\right)^2
}
d\mes(u), 
\end{split}
\end{equation*}
where $\futureCov{u,u}$ denotes the $\no \times \no$ covariance matrix between all $\no$
responses at point $u$ conditional on the first $n+1$ observation batches.
Now, by using co-kriging update formulae and our shortcut notation for the CDF of centred
multivariate Gaussian vectors, we observe that
\begin{equation*}
\begin{split}
&\varPhi_{\no}\left(\bt;~\futureMean{\uu}, \futureCov{u, u}\right) 
\\
=&
\varPhi_{\no}\left(\bt-\futureMean{\uu}; \futureCov{u, u}\right) \\
=&
\varPhi_{\no}\left(\bt-\currentMean{\uu}-\lambda_{[n+1,n+1]}(u)^T(\gp[\bm{x}_{n+1}]
-\currentMean{\bm{x}_{n+1}}), \futureCov{u, u}\right) \\
=&
\varPhi_{\no}\left(a + BV, \futureCov{u, u}\right),
\end{split}
\end{equation*}
with $a=\bt-\currentMean{\uu}$, 
$B=-\lambda_{[n+1,n+1]}(u)^T$
and $V=\gp[\bm{x}_{n+1}]-\currentMean{\bm{x}_{n+1}}$.
Applying Proposition~\ref{propo2} then delivers that
\begin{equation*}
\begin{split}
&\currentExp{
    \varPhi_{\no}\left(\bt;~\futureMean{\uu}, \futureCov{u, u}\right)^2 
}
=\varPhi_{2\no}
\left(
\left(
\begin{matrix}
\bt-\currentMean{\uu}\\
\bt-\currentMean{\uu}
\end{matrix}
\right);
\mathbf{\Sigma}_{[\stage]}(\uu)
\right),
\end{split}
\end{equation*}
with $\mathbf{\Sigma}_{[n]}(\uu)$ as in the formulation of the proposition. This completes the proof.
\end{proof}

\begin{proof}{(Proposition \ref{propo_emv})}
\begin{equation*}
\begin{split}
\currentEEMV(\bm{x})
&=\int_{\domain^2} 
\varPhi_{2\no}
\left(
(\bt, \bt); \mu((u,v)), 
K((u,v),(u,v))
\right) 
\
\mathrm{d}\mes^{\otimes} 
(u,v)\\
&-\int_{\domain^2} 
\currentExp{
\varPhi_{\no}\left(\bt; \futureMean{\uu}, K_{[n+1]}(\uu, \uu)\right)
\varPhi_{\no}\left(\bt; \futureMean{\vv}, K_{[n+1]}(\vv, \vv)\right)
}
\
\mathrm{d}\mes^{\otimes} 
(u,v)
\end{split}
\end{equation*}
and the proof follows by applying Proposition \ref{propo3}
with
$$V=\gp[\bm{x}_{n+1}]-\currentMean{\bm{x}_{n+1}} \sim \mathcal{N}(0_{q_{n+1}},k_{[n]}(\bm{x}_{n+1},\bm{x}_{n+1}))$$
and $a_1=\bt-\currentMean{\uu}$,
$B_1=-\lambda_{[n+1,n+1]}(\uu)^T$, $a_2=\bt-\currentMean{\vv}$, $B_2=-\lambda_{[n+1,n+1]}(\vv)^T$, $C_1 = \currentCov{\uu, \uu}$, $C_2 = \currentCov{\vv, \vv}$.
\end{proof}
\end{document}